\documentclass[a4paper]{article}
\usepackage{amsmath}
\usepackage[english]{babel}
\usepackage{latexsym}
\usepackage{amssymb}
\usepackage{amscd}
\usepackage{amsgen,amstext,amsbsy,amsopn}
\usepackage{math rsfs}
\usepackage{bm,bbm}
\usepackage{amsthm,epsfig,graphicx,graphics}
\usepackage[latin1]{inputenc}
\usepackage{xspace}
\usepackage{amsxtra}
\usepackage{color}
\usepackage{dsfont}

\usepackage[pdftex]{hyperref}

\usepackage{geometry}
\newcommand{\alj}{\alpha_n}
\newcommand{\aljj}{\alpha_{n+1}}
\newcommand{\kj}{k_n}
\newcommand{\glfj}{\mathcal{E}^{\mathrm{GL}}_{n}}
\newcommand{\tj}{\bar{t}_{n,\eps}}
\newcommand{\tjj}{\bar{t}_{n+1,\eps}}
\newcommand{\gjj}{G_{n,n+1}}
\newcommand{\ijj}{I_{n,n+1}}
\newcommand{\ikk}{I_{k,k'}}
\newcommand{\tk}{\bar{t}_{k,\eps}}
\newcommand{\tkk}{\bar{t}_{k',\eps}}
\newcommand{\bteps}{\bar{t}_{\eps}}

\numberwithin{equation}{section}
\newcommand{\bdm}{\begin{displaymath}}
\newcommand{\edm}{\end{displaymath}}
\newcommand{\bdn}{\begin{eqnarray}}
\newcommand{\edn}{\end{eqnarray}}
\newcommand{\bay}{\begin{array}{c}}
\newcommand{\eay}{\end{array}}
\newcommand{\ben}{\begin{enumerate}}
\newcommand{\een}{\end{enumerate}}
\newcommand{\beq}{\begin{equation}}
\newcommand{\eeq}{\end{equation}}
\newcommand{\beqn}{\begin{eqnarray}}
\newcommand{\eeqn}{\end{eqnarray}}
\newcommand{\bml}[1]{\begin{multline} #1 \end{multline}}
\newcommand{\bmln}[1]{\begin{multline*} #1 \end{multline*}}

\newcommand{\lf}{\left}
\newcommand{\ri}{\right}

\newcommand{\rv}{\mathbf{r}}

\newcommand{\aae}{a_{\eps}}

\newcommand{\beps}{b_{k}}
\newcommand{\deps}{\delta_{\eps}}

\newcommand{\de}{d_{\eps}}
\newcommand{\diff}{\mathrm{d}}
\newcommand{\eps}{\varepsilon}

\newcommand{\dist}{\mathrm{dist}}

\newcommand{\ie}{I_{\eps}}

\newcommand{\ieps}{I_{\eps}}
\newcommand{\iepst}{\tilde{I}_{\eps}}
\newcommand{\game}{\gamma_{\eps}}
\newcommand{\sigme}{\sigma_{\eps}}
\newcommand{\varre}{\varrho_{\eps}}

\newcommand{\gav}{\bm{\gamma}}
\newcommand{\nuv}{\bm{\nu}}

\newcommand{\ba}{\mathcal{B}}

\newcommand{\teps}{t_{\eps}}

\newcommand{\gle}{E^{\mathrm{GL}}}
\newcommand{\glm}{\Psi^{\mathrm{GL}}}

\newcommand{\gldom}{\mathscr{D}^{\mathrm{GL}}}
\newcommand{\aav}{\mathbf{A}}
\newcommand{\aavm}{\mathbf{A}^{\mathrm{GL}}}

\newcommand{\hex}{b}
\newcommand{\theo}{\Theta_0}

\newcommand{\glfk}{\mathcal{G}_{\kappa,\sigma}^{\mathrm{GL}}}

\newcommand{\glfe}{\mathcal{G}_{\eps}^{\mathrm{GL}}}
\newcommand{\glee}{E_{\eps}^{\mathrm{GL}}}

\newcommand{\annf}{\mathcal{G}_{\ann}}
\newcommand{\anne}{G_{\ann}}

\newcommand{\curv}{k(s)}

\newcommand{\eones}{E^{\mathrm{1D}}_{\star}}
\newcommand{\fal}{\fkal}

\newcommand{\fred}{\mathcal{F}_{\rm red}}

\newcommand{\pot}{V_{k,\alpha}}
\newcommand{\pots}{W_{\alpha}}

\newcommand{\curl}{\mbox{curl}}

\newcommand{\ann}{\mathcal{A}_{\eps}}
\newcommand{\annt}{\tilde{\mathcal{A}}_{\eps}}
\newcommand{\annd}{\mathcal{A}_{\rm bl}}

\newcommand{\spac}{\ell_{\eps}}
\newcommand{\cell}{\mathcal{C}}

\newcommand{\cellj}{\mathcal{C}_n}

\newcommand{\rest}{\mathcal{R}_n}
\newcommand{\domj}{\mathcal{D}_n}

\newcommand{\half}{\mbox{$\frac{1}{2}$}}

\newcommand{\disp}{\displaystyle}
\newcommand{\tx}{\textstyle}

\newcommand{\neps}{N_{\eps}}

\newcommand{\Z}{\mathbb{Z}}
\newcommand{\R}{\mathbb{R}}
\newcommand{\N}{\mathbb{N}}

\newcommand{\C}{\mathbb{C}}

\newcommand{\E}{\mathcal{E}}

\newcommand{\OO}{\mathcal{O}}

\newcommand{\al}{\alpha}

\newcommand{\ep}{\varepsilon}

\newcommand{\Om}{\Omega}

\newcommand{\dd}{\partial}

\newcommand{\supp}{\mathrm{supp}}

\newcommand{\Hcc}{H_{\mathrm{c}2}}

\newcommand{\logi}{|\log \eps| ^{\infty}}

\newcommand{\fkstar}{f_{k}}
\newcommand{\fkstari}{f_{k_n}}
\newcommand{\fkstariplus}{f_{k_{n+1}}}

\newtheorem{teo}{Theorem}[section]
\newtheorem{lem}{Lemma}[section]
\newtheorem{pro}{Proposition}[section]
\newtheorem{cor}{Corollary}[section]

\newcounter{remark}[section]
\newenvironment{rem}{\refstepcounter{remark} \vspace{0,1cm} \noindent \textit{Remark \thesection.\theremark}\,}{\hfill \qed \vspace{0,2cm}}

\newcommand{\annbk}{\bar{I}_{k,\eps}}

\newcommand{\btik}{\bar{t}_{k,\eps}}

\newcommand{\fkal}{f_{k,\alpha}}
\newcommand{\fk}{f_{k}}

\newcommand{\fO}{f_0}

\newcommand{\fone}{\E^{\mathrm{1D}}}
\newcommand{\fonekal}{\E ^{\rm 1D}_{k,\alpha}}

\newcommand{\fonek}{\E ^{\rm 1D}_{k}}

\newcommand{\eone}{E^{\mathrm{1D}}}
\newcommand{\eonekal}{E ^{\rm 1D}_{k,\alpha}}

\newcommand{\eonek}{E ^{\rm 1D}_{\star} (k)}
\newcommand{\eoneo}{E ^{\rm 1D}_{0}}

\newcommand{\alk}{\alpha(k)}

\newcommand{\alO}{\alpha_0}
\newcommand{\Fk}{F_k}
\newcommand{\Kk}{K_k}

\newcommand{\potkal}{V_{k,\alpha}}

\newcommand{\potk}{V_{k}}

\newcommand{\Sloc}{S_{\rm loc}}
\newcommand{\Sglob}{S_{\rm glo}}
\newcommand{\psit}{\psi_{\rm trial}}

\newcommand{\Gref}{g_{\rm ref}}
\newcommand{\Kti}{\tilde{K}_{n}}

\begin{document}

\markboth{\scriptsize{\textsc{Correggi, Rougerie} -- Surface Superconductivity}}{\scriptsize{\textsc{Correggi, Rougerie} -- Surface Superconductivity}}

\title{Boundary Behavior of the Ginzburg-Landau Order Parameter in the Surface Superconductivity Regime}

\author{M. Correggi${}^{a}$, N. Rougerie${}^{b}$
	\\
	\normalsize\it ${}^{a}$ Dipartimento di Matematica e Fisica, Universit\`{a} degli Studi Roma Tre,	\\
	\normalsize\it L.go San Leonardo Murialdo, 1, 00146, Rome, Italy.	\\
	\normalsize\it ${}^{b}$ Universit\'e de Grenoble 1 \& CNRS, LPMMC \\ 
	\normalsize\it Maison des Magist\`{e}res CNRS, BP166, 38042 Grenoble Cedex, France.}
	
\date{January 11th, 2015}

\maketitle

\begin{abstract} 
We study the 2D Ginzburg-Landau theory for a type-II superconductor in an applied magnetic field varying between the second and third critical value. In this regime the order parameter minimizing the GL energy is concentrated along the boundary of the sample and is well approximated to leading order (in $L ^2$ norm) by a simplified 1D profile in the direction perpendicular to the boundary. Motivated by a conjecture of Xing-Bin Pan, we address the question of whether this approximation can hold uniformly in the boundary region. We prove that this is indeed the case as a corollary of a refined, second order energy expansion including contributions due to the curvature of the sample. Local variations of the GL order parameter are controlled by the second order term of this energy expansion, which allows us to prove the desired uniformity of the surface superconductivity layer. 
\end{abstract}

\tableofcontents

\section{Introduction}

The Ginzburg-Landau (GL) theory of superconductivity, originating in~\cite{GL}, provides a phenomenological, macroscopic, description of the response of a superconductor to an applied magnetic field. Several years after it was introduced, it turned out that it could be derived from the microscopic BCS theory~\cite{BCS,Gor} and should thus be seen as a mean-field/semiclassical approximation of many-body quantum mechanics. A mathematically rigorous derivation starting from BCS theory has been provided recently~\cite{FHSS}. 

Within GL theory, the state of 	a superconductor is described by an order parameter $\Psi:~\R ^2\to\C$ and an induced magnetic vector potential $\kappa \sigma \aav:\R ^2 \to \R ^2 $ generating an induced magnetic field 
$$h=\kappa \sigma  \: \curl \, \aav.$$
The ground state of the theory is found by minimizing the energy functional\footnote{Here we use the units of~\cite{FH-book}, other choices are possible, see, e.g., \cite{SS2}.}
\beq\label{eq:gl func}
	\glfk[\Psi,\aav] = \int_{\Om} \diff \rv \: \lf\{ \lf| \lf( \nabla + i \kappa\sigma  \aav \ri) \Psi \ri|^2 - \kappa^2 |\Psi|^2 + \half \kappa^2 |\Psi|^4 + \lf(\kappa \sigma \ri)^2 \lf| \curl \aav - 1 \ri|^2 \ri\},
\eeq
where $ \kappa >0  $ is a physical parameter (penetration depth) characteristic of the material, and $\kappa \sigma  $ measures the intensity of the external magnetic field, that we assume to be constant throughout the sample. We consider a model for an infinitely long cylinder of cross-section $\Om \subset \R ^2$, a compact simply connected set with regular boundary.

Note the invariance of the functional under the gauge transformation 
\begin{equation}\label{eq:gauge inv}
\Psi \to \Psi e^{-i \kappa \sigma \varphi}, \qquad \aav \to \aav + \nabla \varphi, 
\end{equation}
which implies that the only physically relevant quantities are the gauge invariant ones such as the induced magnetic field $h$ and the density $|\Psi| ^2$. The latter gives the local relative density of electrons bound in Cooper pairs. It is well-known that a minimizing $\Psi$ must satisfy $|\Psi| ^2 \leq 1$. A~value $|\Psi| = 1$ (respectively, $|\Psi| = 0$) corresponds to the superconducting (respectively, normal) phase where all (respectively, none) of the electrons form Cooper pairs. The perfectly superconducting state with $|\Psi| = 1$ everywhere is an approximate ground state of the functional for small applied field and the normal state where $\Psi$ vanishes identically is the ground state for large magnetic field. In between these two extremes, different mixed phases can occur, with normal and superconducting regions varying in proportion and organization.

A vast mathematical literature has been devoted to the study of these mixed phases in type-II superconductors (characterized by $\kappa > 1/\sqrt{2}$), in particular in the limit $\kappa \to \infty$ (extreme type-II). Reviews and extensive lists of references may be found in~\cite{FH-book,SS2,Sig}. Two main phenomena attracted much attention:
\begin{itemize}
\item The formation of hexagonal vortex lattices when the applied magnetic field varies between the first and second critical field, first predicted by Abrikosov~\cite{Abr}, and later experimentally observed (see, e.g., \cite{Hetal}). In this phase, vortices (zeros of the order parameter with quantized phase circulation) sit in small normal regions included in the superconducting phase and form regular patterns.  
\item The occurrence of a surface superconductivity regime when the applied magnetic fields varies between the second and third critical fields. In this case, superconductivity is completely destroyed in the bulk of the sample and survives only at the boundary, as predicted in~\cite{SJdG}. We refer to~\cite{NSG} for experimental observations. 
\end{itemize}
We refer to~\cite{CR2} for a more thorough discussion of the context. We shall be concerned with the surface superconductivity regime, which in the above units translates into the assumption 
\beq
	\label{eq:external field}
	\sigma = b \kappa
\eeq
for some fixed parameter $b$ satisfying the conditions
\beq
	\label{eq:b condition}
	1 < b < \theo^{-1}
\eeq
where $\theo$ is a spectral parameter (minimal ground state energy of the shifted harmonic oscillator on the half-line, see~\cite[Chapter 3]{FH-book}):
\begin{equation}\label{eq:theo}
\theo := \inf_{\alpha \in \R} \inf \left\{ \int_{\R ^+} \diff t \left( |\dd_t u| ^2 + (t+\alpha) ^2 |u| ^2 \right), \: \left\Vert u \right\Vert_{L^2 (\R ^+)} = 1 \right\}.
\end{equation}
From now on we introduce more convenient units to deal with the surface superconductivity phenomenon: we define the small parameter 
\beq
	\label{eq:eps}
	\eps = \frac{1}{\sqrt{\sigma \kappa}} = \frac{1}{b^{1/2}\kappa} \ll 1 
\eeq
and study the asymptotics $ \eps \to 0 $ of the minimization of the functional~\eqref{eq:gl func}, which in the new units reads
\beq\label{eq:GL func eps}
	\glfe[\Psi,\aav] = \int_{\Om} \diff \rv \: \bigg\{ \bigg| \bigg( \nabla + i \frac{\aav}{\eps^2} \bigg) \Psi \bigg|^2 - \frac{1}{2 \hex \eps^2} \lf( 2|\Psi|^2 - |\Psi|^4 \ri) + \frac{\hex}{\eps^4} \lf| \curl \aav - 1 \ri|^2 \bigg\}.
\eeq
We shall denote
\beq
	\glee : = \min_{(\Psi, \aav) \in \gldom} \glfe[\Psi,\aav],
\eeq
with
\beq
	\gldom : = \lf\{ (\Psi,\aav) \in H^1(\Om;\C) \times H^1(\Om;\R^2) \ri\},
\eeq
and denote by $ (\glm,\aavm) $ a minimizing pair (known to exist by standard methods~\cite{FH-book,SS2}). 

\medskip

The salient features of the surface superconductivity phase are as follows:
\begin{itemize}
\item The GL order parameter is concentrated in a thin boundary layer of thickness $\sim \eps = (\kappa \sigma) ^{-1/2}$. It decays exponentially to zero as a function of the distance from the boundary.
\item The applied magnetic field is very close to the induced magnetic field, $\curl \, \aav \approx 1$.
\item Up to an appropriate choice of gauge and a mapping to boundary coordinates, the ground state of the theory is essentially governed by the minimization of a 1D energy functional in the direction perpendicular to the boundary. 
\end{itemize}
A review of rigorous statements corresponding to these physical facts may be found in~\cite{FH-book}. One of their consequences is the energy asymptotics 
\beq
	\label{eq:FH energy asympt 2}
	\glee = \frac{|\dd \Om| \eoneo}{\eps} + \OO(1),
\eeq
where $|\dd\Om|$ is the length of the boundary of $\Om$, and $\eoneo$ is obtained by minimizing the functional
\begin{equation}\label{eq:intro 1D func}
\fone_{0,\alpha}[f] : = \int_0^{+\infty} \diff t \lf\{ \lf| \partial_t f \ri|^2 + (t + \alpha )^2 f^2 - \frac{1}{2b} \lf(2 f^2 - f^4 \ri) \ri\},
\end{equation}
both with respect to the function $f$ and the real number $\alpha$. We proved recently~\cite{CR2} that~\eqref{eq:FH energy asympt 2} holds in the full surface superconductivity regime, i.e. for $1<b<\theo ^{-1}$. This followed a series of partial results due to several authors~\cite{Alm,AH,FH1,FH2,FHP,LP,Pan}, summarized in~\cite[Theorem 14.1.1]{FH-book}. Some of these also concern the limiting regime $b \nearrow \theo ^{-1}$. The other limiting case $b\searrow 1$ where the transition from boundary to bulk behavior occurs is studied in~\cite{FK,Kac}, whereas results in the regime $b\nearrow 1$ may be found in~\cite{AS,Alm2,SS1}.

The rationale behind~\eqref{eq:FH energy asympt 2} is that, up to a suitable choice of gauge, any minimizing order parameter $\glm$ for~\eqref{eq:gl func} has the structure 
\begin{equation}\label{eq:GLm structure formal}
\glm(\rv) \approx \fO \left(\tx \frac{\tau}{\eps} \right)  \exp \left( - i \alO \tx \frac{s}{\eps}\right) \exp \lf\{ i \phi_{\eps}(s,t) \ri\}
\end{equation}
where $(\fO,\alO)$ is a minimizing pair for~\eqref{eq:intro 1D func}, $(s,\tau)= $ (tangent coordinate, normal coordinate) are boundary coordinates defined in a tubular neighborhood of $\dd \Om$ with $ \tau = \dist(\rv,\partial \Omega) $ for any point $ \rv $ there and $ \phi_{\eps} $ is a gauge phase factor (see \eqref{eq: gauge phase}), which plays a role in the change to boundary coordinates. Results in the direction of~\eqref{eq:GLm structure formal} may be found in the {following references:}
\begin{itemize}
\item {\cite{Pan} contains a result of uniform distribution of the energy density at the domain's boundary for any $ 1 \leq b < \theo^{-1} $};
\item {\cite{FH1} gives fine energy estimates compatible with~\eqref{eq:GLm structure formal} when $ b \nearrow \theo^{-1} $};
\item {\cite{AH} and then~\cite{FHP} prove that~\eqref{eq:GLm structure formal} holds at the level of the density, in the $L ^2$ sense, for $1.25\leq b < \theo ^{-1}$;}
\item {\cite{FK} and then~\cite{Kac} investigate the concentration of the energy density when $b$ is close to $1$;}
\item {\cite{FKP} contains results about the energy concentration phenomenon in the 3D case.}
\end{itemize}
In~\cite[Theorem~2.1]{CR2} we proved that 
\begin{equation}\label{eq:recall density generic}
\left\Vert |\glm| ^2 -  \fO ^2 \left(\tx\frac{\tau}{\eps} \right)  \right\Vert_{L ^2 (\Om)} \leq C \eps \ll \left\Vert \fO ^2 \left(\tx\frac{\tau}{\eps}\right)  \right\Vert_{L ^2 (\Om)} 
\end{equation}
for any $1<b<\theo ^{-1}$ in the limit $\eps \to 0$. A very natural question is whether the above estimate may be improved to a uniform control (in $L ^{\infty}$ norm) of the local discrepancy between the modulus of the true GL minimizer and the simplified normal profile $\fO \left(\tx\frac{\tau}{\eps} \right)$. Indeed,~\eqref{eq:recall density generic} is still compatible with the vanishing of $\glm$ in small regions, e.g., vortices, inside of the boundary layer. Proving that such local deviations from the normal profile do not occur would explain the observed uniformity of the surface superconducting layer (see again~\cite{NSG} for experimental pictures). Interest in this problem (stated as Open Problem number 4 in the list in \cite[Page 267]{FH-book}) originates from a conjecture of X.B. Pan~\cite[Conjecture 1]{Pan} and an affirmative solution has been provided in~\cite{CR2} for the particular case of a disc sample. The purpose of this paper is to extend the result to general  samples with regular 
boundary (the case with corners is known to require a different analysis~\cite[Chapter 15]{FH-book}). 

Local variations (on a scale $\OO (\eps)$) in the tangential variable are compatible with the energy estimate~\eqref{eq:FH energy asympt 2}, and thus the uniform estimate obtained for disc samples in~\cite{CR2} is based on an expansion of the energy to the next order:  
\beq
\label{eq:recall energy disc}
\glee = \frac{2\pi \eonek}{\eps} + \OO(\eps |\log\eps|),
\eeq
where $\eonek$ is the minimum (with respect to both the real number $\alpha$ and the function $f$) of the $\eps$-dependent functional
\begin{equation}
\label{eq:intro 1D func disc}
\fone_{k,\alpha}[f] : = \int_0^{c_0|\log\eps|} \diff t \: (1-\eps k t )\lf\{ \lf| \partial_t f \ri|^2 + \frac{(t + \alpha - \frac12 \eps k t ^2 )^2}{(1-\eps k t ) ^2} f^2 - \frac{1}{2b} \lf(2 f^2 - f^4 \ri) \ri\},
\end{equation}
where the constant $ c_0 $  has to be chosen large enough and $k=R ^{-1}$ is the curvature of the disc under consideration, whose radius we denote by $R$. Of course, \eqref{eq:intro 1D func} is simply the above functional where one sets $k=0$, $\eps=0$, which amounts to neglect the curvature of the boundary. When the curvature is constant,~\eqref{eq:recall energy disc} in fact follows from a next order expansion of the GL order parameter beyond~\eqref{eq:GLm structure formal}:
\begin{equation}\label{eq:GLm structure formal disc}
\glm(\rv) \approx \fk \left(\tx \frac{\tau}{\eps} \right)  \exp \left( - i \alk \tx \frac{s}{\eps}\right) \exp \lf\{ i \phi_{\eps}(s,t) \ri\}
\end{equation}
where $(\alk,\fk)$ is a minimizing pair for~\eqref{eq:intro 1D func disc}. Note that for any fixed $k$
\beq
	\label{eq:point est 0 profile}
	 \fk = \fO (1+\OO(\eps)),\qquad \alk = \alO (1+\OO(\eps)),
\eeq
so that~\eqref{eq:GLm structure formal disc} is a slight refinement of~\eqref{eq:GLm structure formal} but the $\OO(\eps)$ correction corresponds to a contribution of order $1$ beyond~\eqref{eq:FH energy asympt 2} in~\eqref{eq:recall energy disc}, which turns out to be the order that controls local density variations. 

As suggested by the previous results in the disc case, the corrections to the energy asymptotics~\eqref{eq:FH energy asympt 2} must be curvature-dependent. The case of a general sample where the curvature of the boundary is not constant is then obviously harder to treat than the case of a disc, where one obtains~\eqref{eq:recall energy disc} by a simple variant of the proof of~\eqref{eq:FH energy asympt 2}, as explained in our previous paper~\cite{CR2}. 

In fact, we shall obtain below the desired uniformity result for the order parameter in general domains as a corollary of the energy expansion ($\gamma$ is a fixed constant)
\begin{equation}\label{eq:intro energy GL}
\boxed{\glee = \frac{1}{\eps} \int_0^{|\partial \Omega|} \diff s \: \eone_\star \left(k(s)\right)  + \OO (\eps |\log \eps| ^\gamma) }
\end{equation}
where the integral runs over the boundary of the sample, $k(s)$ being the curvature of the boundary as a function of the tangential coordinate $s$. Just as the particular case~\eqref{eq:recall energy disc},~\eqref{eq:intro energy GL} contains the leading order~\eqref{eq:FH energy asympt 2}, but $\OO(1)$ corrections are also evaluated precisely. As suggested by the energy formula, the GL order parameter has in fact small but fast variations in the tangential variable which contribute to the subleading order of the energy. More precisely, one should think of the order parameter as having the approximate form 
\begin{equation}\label{eq:intro GLm formal refined}
\boxed{\glm(\rv) = \glm (s,\tau) \approx f_{k(s)} \left(\tx \frac{\tau}{\eps} \right)  \exp \left( - i \alpha (k(s)) \tx \frac{s}{\eps}\right) \exp\lf\{i \phi_{\eps}(s,t) \ri\} }
\end{equation}
with $f_{k(s)},\alpha(k(s))$ a minimizing pair for the energy functional~\eqref{eq:intro 1D func disc} at curvature $k = k(s)$. The main difficulty we encounter in the present paper is to precisely capture the subtle curvature dependent variations encoded in~\eqref{eq:intro GLm formal refined}. What our new result (we give a rigorous statement below)~\eqref{eq:intro GLm formal refined} shows is that curvature-dependent deviations to~\eqref{eq:GLm structure formal} do exist but are of limited amplitude and can be completely understood via the minimization of the family of 1D functionals~\eqref{eq:intro 1D func disc}.  A crucial input of our analysis is therefore a detailed inspection of the $k$-dependence of the ground state of~\eqref{eq:intro 1D func disc}. 

We can deduce from~\eqref{eq:intro energy GL} a uniform density estimate settling the general case of~\cite[Conjecture 1]{Pan} and~\cite[Open Problem 4, page 267]{FH-book}. We believe that the energy estimate~\eqref{eq:intro energy GL} is of independent interest since it helps in clarifying the role of domain curvature in surface superconductivity physics. It was previously known (see \cite[Chapters 8 and 13]{FH-book} and references therein) that corrections to the value of the third critical field depend on the domain's curvature, but applications of these results are limited to the regime where $b\to \theo ^{-1}$ when $\eps \to 0$. The present paper seems to contain the first results indicating the role of the curvature in the regime~$1<b<\theo ^{-1}$. This role may seem rather limited since it only concerns the second order in the energy asymptotics but it is in fact crucial in controlling local variations of the order parameter and allowing to prove a strong form of  uniformity for the surface 
superconductivity layer.

\medskip

Our main results are rigorously stated and further discussed in Section~\ref{sec:main results}, their proofs occupy the rest of the paper. Some material from~\cite{CR2} is recalled in Appendix~\ref{sec:app} for convenience.

\medskip

\noindent\textbf{Notation.} In the whole paper, $C$ denotes a generic fixed positive constant independent of $\eps$ whose value changes from formula to formula. A $\OO (\delta)$ is always meant to be a quantity whose absolute value is bounded by $\delta = \delta (\eps)$ in the limit $\eps \to 0$. We use $\OO (\eps ^{\infty})$ to denote a quantity (like $\exp(- \eps ^{-1})$) going to $0$ faster than any power  of $\eps$ and $\logi$ to denote $|\log \eps|^a$ where $a>0$ is some unspecified, fixed but possibly large constant. Such quantities will always appear multiplied by a power of $\eps$, e.g., $\eps \logi$ which is a $\OO (\eps ^{1-c})$ for any $0<c<1$, and hence we usually do not specify the precise power $a$.

\medskip

\noindent\textbf{Acknowledgments.} M.C. acknowledges the support of MIUR through the FIR grant 2013 ``Condensed Matter in Mathematical Physics (Cond-Math)'' (code RBFR13WAET). N.R. acknowledges the support of the ANR project Mathostaq (ANR-13-JS01-0005-01). We also acknowledge the hospitality of the \emph{Institut Henri Poincar\'e}, Paris. {We are indebted to one of the anonymous referees for the content of Remarks~ 2.\ref{rem:b1} and 2.\ref{rem:curvature}.}

\section{Main Results}\label{sec:main results}

\subsection{Statements}\label{sec:statements}

We first state the refined energy and density estimates that reveal the contributions of the domain's boundary. As suggested by~\eqref{eq:intro GLm formal refined}, we now introduce a reference profile that includes these variations. A piecewise constant function in the tangential direction is sufficient for our purpose and we thus first introduce a decomposition of the superconducting boundary layer that will be used in all the paper. The thickness of this layer in the normal direction should roughly be of order $\eps$, but to fully capture the phenomenon at hand we need to consider a layer of size $c_0 \eps |\log \eps|$ where $c_0$ is a fixed, large enough constant. By a passage to boundary coordinates and dilation of the normal variable on scale~$\eps$ (see~\cite[Appendix F]{FH-book} or Section~\ref{sec:up bound} below), the surface superconducting layer 
\beq
	\label{eq:intro ann}
	 \annt : = \lf\{ \rv \in \Omega \: | \: \tau \leq c_0 \eps |\log\eps| \ri\},
\eeq
where
\beq
	\tau : = \dist(\rv, \partial \Omega),
\eeq
can be mapped to 
\begin{equation}\label{eq:intro def ann rescale}
\ann:= \left\{ (s,t) \in \left[0, |\partial \Omega| \right] \times \left[0,c_0 |\log\eps|\right] \right\}.
\end{equation}
We split this domain into $ \neps = \OO(\eps ^{-1})$ rectangular cells $ \{ \cellj \}_{n=1, \ldots, \neps}$ of constant side length $ \spac \propto \eps$ in the $s$ direction. We denote $s_n $, $s_{n+1} = s_n + \spac $ the $s$ coordinates of the boundaries of the cell $ \cellj $:
$$ \cellj = [s_n,s_{n+1}] \times [0,c_0 |\log \eps|]$$
and we may clearly choose 
$$ \spac = \eps |\dd \Om|\lf(1 + \OO(\eps)\ri)$$
for definiteness. We will approximate the curvature $k(s)$ by its mean value $k_n$ in each cell:
$$ k_n := \spac^{-1} \int_{s_n} ^{s_{n+1}} \diff s \, k(s).$$
We also denote 
$$  f_n := f_{k_n}, \qquad \alpha_n := \alpha(k_n)$$
respectively the optimal profile and phase associated to $k_n$, obtained by minimizing~\eqref{eq:intro 1D func disc} first with respect to\footnote{We are free to impose $f_n \geq 0$, which we always do in the sequel.} $f$ and then to $\alpha$. 

The reference profile is then the piecewise continuous function
\begin{equation}\label{eq:ref profile}
\Gref (s,t) := 	f_n (t), 	\qquad	\mbox{for } s\in [s_n,s_{n+1}] \mbox{ and } (s,t) \in \ann,
\end{equation}
 that can be extended to the whole domain $ \Omega $ by setting it equal to $ 0 $ for $ \dist(\rv, \partial \Omega) \geq c_0 \eps |\log\eps| $. We compare the density of the full GL order parameter to $\Gref $ in the next theorem. Note that because of the gauge invariance of the energy functional, the phase of the order parameter is not an observable quantity, so the next statement is only about the density $|\glm| ^2$. 

	\begin{teo}[\textbf{Refined energy and density asymptotics}]\label{theo:energy}\mbox{}\\
		Let $\Om\subset \R ^2$ be any smooth, bounded and simply connected domain. For any fixed $1<b<\theo ^{-1}$, in the limit $ \eps \to 0$, it holds
		\begin{equation}\label{eq:energy GL}
			\glee = \frac{1}{\eps} \int_0^{|\partial \Omega|} \diff s \: \eone_\star \left(k(s)\right)  + \OO (\eps \logi). 
		\end{equation}
		and
		\begin{equation}\label{eq:main density}
			\left\Vert |\glm| ^2 -  \Gref ^2 \left(s,\eps ^{-1}  \tau \right)  \right\Vert_{L ^2 (\Om)} = \OO(\eps ^{3/2} \logi) \ll \left\Vert \Gref ^2 \left(s,\eps ^{-1} t \right)  \right\Vert_{L ^2 (\Om)}. 
		\end{equation}
	\end{teo}

\begin{rem}[The energy to subleading order]\label{rem:energy sublead}\mbox{}\\
The most precise result prior to the above is~\cite[Theorem 2.1]{CR2} where the leading order is computed and the remainder is shown to be at most of order $1$. Such a result had been obtained before in~\cite{FHP} for a smaller range of parameters, namely for $1.25\leq b <\theo ^{-1}$, see also~\cite[Chapter~14]{FH-book} and references therein.  The above theorem evaluates precisely the $\OO(1)$ term, which is better appreciated in  light of the following comments:
\begin{enumerate}
\item In the effective 1D functional~\eqref{eq:intro 1D func disc}, the parameter $k$ that corresponds to the local curvature of the sample appears with an $\eps$ prefactor. As a consequence, one may show (see Section~\ref{sec:eff func} below) that for all $s\in [0,|\dd \Om|]$
\beq
	\label{eq:eone asympt}
	\eone_\star \left(k(s)\right) = \eone_\star (0) + \OO(\eps)
\eeq
so that~\eqref{eq:energy GL} contains the previously known results. More generally we prove below that
$$\left|\eone_\star \left(k(s)\right) - \eone_\star \left(k(s')\right)\right| \leq C  \eps |s-s'| $$
so that $\eone_\star \left(k(s)\right)$ has variations of order $\eps$ on the scale of the boundary layer. These contribute to a term of order $1$ that is included in~\eqref{eq:energy GL}. Actually one could investigate the asymptotics \eqref{eq:eone asympt} further, aiming at evaluating explicitly the error $ \OO(\eps) $ and therefore the curvature contribution to the energy. This would in particular be crucial in the analysis described in Remark 2.\ref{rem:curvature} below, but we do not pursue it here for the sake of brevity.

\item Undoing the mapping to boundary coordinates, one should note that $\Gref(s,\eps ^{-1} t)$ has fast variations (at scale $\eps$) in both the $t$ direction and $s$ directions. The latter are of limited amplitude however, which explains that they enter the energy only at subleading order, and why a piecewise constant profile is sufficient to capture the physics.
\item We had previously proved the density estimate~\eqref{eq:recall density generic}, which is less precise than~\eqref{eq:main density}. Note in particular that~\eqref{eq:main density} does not hold at this level of precision if one replaces $\Gref ^2 \left(s,\eps ^{-1} t\right)$ by  the simpler profile $\fO ^2 (\eps^{-1} t)$.
\item Strictly speaking the function $ \Gref $ is defined only in the boundary layer $ \annt  $, so that \eqref{eq:main density} should be interpreted as if $ \Gref $ would vanish outside $ \annt $. However the estimate there is obviously true thanks to the exponential decay of $ \glm $.
\end{enumerate}
\end{rem}

\begin{rem}[Regime $ b \to  1 $]
\label{rem:b1}
\mbox{}	\\
A simple inspection of the proof reveals that some of the crucial estimates still hold true even if $ b \to  1 $, where surface superconductivity is also present (see~\cite{Alm,Pan,FK}). The main reason for assuming $b>1$ is that we rely on some well-known decay estimates for the order parameter (Agmon estimates), which hold only in this case. When $ b \to 1 $ one can indeed find suitable adaptations of those estimates (see, e.g., \cite[Chapter 12]{FH-book}), which however make the analysis much more delicate. In particular the positivity of the cost function (Lemma \ref{lem:K positive} in Section \ref{sec:app cost}) heavily relies on the assumption $ b > 1 $ and, although it is expected to be true even if $ b \to  1 $, its proof requires some non-trivial modifications. Moreover while for $ b \geq 1 $ only surface superconductivity is present and our strategy has good chances to work, on the opposite, when $ b \nearrow 1 $, a bulk term appears in the energy asymptotics \cite{FK} and the problem becomes much more subtle.
\end{rem}

We now turn to the uniform density estimates that follow from the above theorem. Here we can be less precise than before. Indeed, as suggested by the previous discussion, a density deviation of order $\eps$ on a length scale of order $\eps$ only produces a $\OO(\eps ^2)$ error in the energy. Thus, using~\eqref{eq:energy GL} we may only rule out local variations of a smaller order than the tangential variations included in~\eqref{eq:ref profile}, and for this reason we will compare $|\glm|$ in $L^{\infty}$ norm only to the simplified profile $\fO (\eps ^{-1}\tau)$, since by \eqref{eq:point est 0 profile} $ \fO(t) - f_k(t) = \OO(\eps) $. Also, the result may be proved only in a region where the density is relatively large\footnote{Recall that it decays exponentially far from the boundary.}, namely in 
\beq
\label{eq:annd}
\annd : = \lf\{ \rv \in \Om \: : \: \fO \lf(\eps ^{-1} \tau \ri) \geq \game \ri\} \subset \lf\{ \dist(\rv,\dd \Om) \leq \half \eps \sqrt{|\log\game|} \ri\},
\eeq
where $\rm{bl}$ stands for ``boundary layer'' and $ 0 < \game \ll 1 $ is any quantity such that
\beq
\label{eq:game}
\game \gg \eps^{1/6}|\log \eps| ^{a},
\eeq
where $ a > 0 $ is a suitably large constant related\footnote{Assuming that \eqref{eq:energy GL} holds true with an error of order $ \eps |\log\eps|^{\gamma} $, for some given $ \gamma > 0 $, the constant $ a $ can be any number satisfying $ a > \frac{1}{6} (\gamma+3) $.} to the power of $ |\log\eps| $ appearing in \eqref{eq:energy GL}.
The inclusion in \eqref{eq:annd} follows from \eqref{eq:fal point l u b} below and ensures we are really considering a significant boundary layer: recall that the physically relevant region has a thickness roughly of order $\eps|\log\eps|$.

\begin{teo}[\textbf{Uniform density estimates and Pan's conjecture}]\label{theo:Pan}\mbox{}\\
		Under the assumptions of the previous theorem, it holds
		\beq
			\label{eq:Pan plus}
			\lf\| \lf|\glm(\rv)\ri| - \fO \lf(\eps ^{-1} \tau \ri) \ri\|_{L^{\infty}(\annd)} \leq C \game^{-3/2} \eps^{1/4} \logi  \ll 1.
		\eeq
		In particular for any $ \rv  \in \partial \Omega $ we have
		\begin{equation}\label{eq:Pan}
			\lf| \lf| \glm(\rv) \ri| - \fO (0) \ri| \leq  C \eps^{1/4} |\logi |\ll 1,
		\end{equation}
		where $C$ does not depend on $\rv$.
	\end{teo}

Estimate~\eqref{eq:Pan} solves the original form of Pan's conjecture~\cite[Conjecture 1]{Pan}. In addition, since $\fO$ is strictly positive, the stronger estimate~\eqref{eq:Pan plus} ensures that $\glm$ does not vanish in the boundary layer~\eqref{eq:annd}. A physical consequence of the theorem is thus that normal inclusions such as vortices in the surface superconductivity phase may not occur. This is very natural in view of the existing knowledge on type-II superconductors but had not been proved previously.  

\medskip

We now return to the question of the phase of the order parameter. Of course, the full phase cannot be estimated because of gauge invariance but gauge invariant quantities linked to the phase can. One such quantity is the winding number (a.k.a. phase circulation or topological degree) of $\glm$ around the boundary $\dd \Om$ defined as
\beq\label{eq:GL degree}
		2 \pi \deg\lf(\Psi, \partial \Om\ri) : = - i \int_{\partial \Om} \diff s \: \frac{|\Psi|}{\Psi} \partial_{s} \lf( \frac{\Psi}{|\Psi|} \ri),
	\eeq
$ \partial_{s} $ standing for the tangential derivative. Theorem~\ref{theo:Pan} ensures that $\deg\lf(\Psi, \partial \ba_R\ri)\in \Z$ is well-defined. Roughly, this quantity measures the number of quantized phase singularities (vortices) that $\glm$ has inside $\Om$. Our estimate is as follows:

	\begin{teo}[{\bf Winding number of $ \glm $ on the boundary}]
		\label{theo:circulation}
		\mbox{}	\\
		Under the previous assumptions, any GL minimizer $ \glm $ satisfies
		\beq\label{eq:GL degree result}
			\deg\lf(\glm, \partial \Omega\ri) = \frac{|\Omega|}{\eps^2} + \frac{|\alO|}{\eps} + \OO(\eps^{-3/4}|\log\eps|^{\infty})
		\eeq
		in the limit $\eps \to 0$.
	\end{teo}

Note that the remainder term in~\eqref{eq:GL degree result} is much larger than $\eps ^{-1}|\alk -\alO| = \OO (1)$ so that the above result does not allow to estimate corrections due to curvature. We believe that, just as we had to expand the energy to second order to obtain the refined first order results Theorems~\ref{theo:Pan} and~\ref{theo:circulation}, obtaining uniform density estimates and degree estimates at the second order would require to expand the energy to the third order, which goes beyond the scope of the present paper.

We had proved Theorems~\ref{theo:Pan} and~\ref{theo:circulation} before in the particular, significantly easier, case where $\Om$ is a disc. The next subsection contains a sketch of the proof of the general case, where new ingredients enter, due to the necessity to take into account the non-trivial curvature of the boundary. {Before proceeding, we make a last remark in this direction:}

\begin{rem}[Curvature dependence of the order parameter]\label{rem:curvature}\mbox{}\\
{In view of previous results~\cite{FH1} in the regime $b \nearrow\theo ^{-1}$, a larger curvature should imply a larger local value of the order parameter. In the regime of interest to this paper, this will only be a subleading order effect, but it would be interesting to capture it by a rigorous asymptotic estimate.} 

{It has been proved before~\cite{Pan,FK} that in the surface superconductivity regime~\eqref{eq:b condition} 
\begin{equation}\label{eq:curv dep psi}
 \frac{1}{b ^{1/2}\eps} |\glm| ^4 \diff \rv  \underset{\eps \to 0}{\longrightarrow}  C(b) \diff s(\rv) 
\end{equation}
as measures, with $ \diff \rv $ the Lebesgue measure and $\diff s (\rv)$ the 1D Hausdorff measure along the boundary. Here $C(b)>0$ is a constant which depends only on $b$. A natural conjecture is that one can derive a result revealing the next-order behavior, of the form
\begin{equation}\label{eq:curv dep psi 2}
\frac{1}{\eps}\left( \frac{1}{b ^{1/2}\eps} |\glm| ^4 \diff \rv -   C(b) \diff s(\rv)\right) \underset{\eps \to 0}{\longrightarrow} C_2 (b)  k(s) \diff s(\rv) 
\end{equation}
with $C_2 (b) >0$ depending only on $b$. The form of the right-hand side is motivated by two considerations: }
\begin{itemize}
\item {In view of~\cite{FH1} we should expect that increasing $k$ increases the local value of $|\glm|$, whence the sign of the correction;} 
\item {Since the curvature appears only at subleading order in this regime, perturbation theory suggests that the correction should be linear in the curvature.} 
\end{itemize}
{We plan to substantiate this picture further in a later work.}
\end{rem}

\subsection{Sketch of proof}\label{sec:sketch}

In the regime of interest to this paper, the GL order parameter is concentrated along the boundary of the sample and the induced magnetic field is extremely close to the applied one. The tools allowing to prove these facts are well-known and described at length in the monograph~\cite{FH-book}. We shall thus not elaborate on this and the formal considerations presented in this subsection take as starting point the following effective functional
\begin{multline}\label{eq:intro GL func bound}
	\annf[\psi] : = \int_0^{|\partial \Omega|} \diff s \int_0^{c_0 |\log\eps|} \diff t \lf(1 - \eps \curv t \ri) \lf\{ \lf| \partial_t \psi \ri|^2 + \frac{1}{(1- \eps \curv t)^2} \lf| \lf( \eps \partial_s + i \aae(s,t) \ri) \psi \ri|^2 \ri.	\\	
	\lf. - \frac{1}{2 \hex} \lf[ 2|\psi|^2 - |\psi|^4 \ri]  \ri\},
\end{multline}
where $(s,t)$ represent boundary coordinates in the original domain $\Om$, the normal coordinate $t$ having been dilated on scale $\eps$, and $ \psi $ can be thought of as $ \glm(\rv(s,\eps t)) $, i.e., the order parameter restricted to the boundary layer. We denote $k(s)$ the curvature of the original domain and have set
\beq\label{eq:intro vect pot bound}
	\aae(s,t) : =- t + \half \eps \curv t^2 + \eps \deps , 	
\eeq
with
\beq\label{eq:intro deps}
	\deps : = \frac{\gamma_0}{\eps^2} - \lf\lfloor \frac{\gamma_0}{\eps^2} \ri\rfloor,	\qquad	 \gamma_0 : = \frac{1}{|\partial \Omega|} \int_{\Omega} \diff \rv \: \curl \, \aavm,
\eeq
$ \lf\lfloor \: \cdot \: \ri\rfloor $ standing for the integer part. Note that a specific choice of gauge has been made to obtain~\eqref{eq:intro GL func bound}. 

Thanks to the methods exposed in~\cite{FH-book}, one can show that the minimization of the above functional gives the full GL energy in units of  $\eps ^{-1}$, up to extremely small remainder terms, provided $c_0$ is chosen lare enough. To keep track of the fact that the domain $\ann = [0, |\partial \Omega|] \times [ 0, c_0 |\log\eps| ]$ corresponds to the unfolded boundary layer of the original domain and $\psi$ to the GL order parameter in boundary coordinates, one should impose periodicity of $\psi$ in the $s$ direction.

Here we shall informally explain the main steps of the proof that 
\begin{equation}\label{eq:intro anne}
\anne = \int_0^{|\partial \Omega|} \diff s \: \eone_\star \left(k(s)\right)  + \OO (\eps ^2 \logi). 
\end{equation}
where $\anne$ is the ground state energy associated to~\eqref{eq:intro GL func bound}. 
When $k(s)\equiv k$ is constant (the disc case), one may use the ansatz
\begin{equation}\label{eq:intro ansatz 1D}
\psi (s,t) = f(t) e^{- i  \left( \eps ^{-1} \alpha s -  \eps \deps s \right) }.  
\end{equation}
and recover the functional~\eqref{eq:intro 1D func disc}. It is then shown in~\cite{CR2} that the above ansatz is essentially optimal if one chooses $\alpha = \alk$ and $f=\fk$. An informal sketch of the proof in the case $k=0$ is given in Section~3.2 therein. The main insight in the general case is to realize that the above ansatz stays valid \emph{locally in $s$}. Indeed, since the terms involving $k(s)$ in~\eqref{eq:intro GL func bound} come multiplied by an $\eps$ factor, it is natural to expect variations in $s$ to be weak and the state of the system to be roughly of the form~\eqref{eq:intro GLm formal refined}, directly inspired by~\eqref{eq:intro ansatz 1D}. 

As usual the upper and lower bound inequalities in~\eqref{eq:intro anne} are proved separately.

\paragraph*{Upper bound.} To recover the integral in the energy estimate~\eqref{eq:intro anne}, we use a Riemann sum over the cell decomposition $\ann = \bigcup_{n=1} ^{\neps} \cellj $ introduced at the beginning of Section~\ref{sec:statements}. Indeed, as already suggested in~\eqref{eq:ref profile}, a piecewise constant approximation in the $s$-direction will be sufficient. Our trial state roughly has the form
\begin{equation}\label{eq:rough trial state}
\psi (s,t) = f_n (t)  e^{- i  \left( \eps ^{-1} \alpha_n s -  \eps \deps s \right) },	\quad \mbox{ for } s_n\leq s \leq s_{n+1}.
\end{equation}
Of course, we need to make this function continuous to obtain an admissible trial state, and we do so by small local corrections, described in more details in Section~\ref{sec:trial state}. We may then approximate the curvature by its mean value in each cell, making a relative error of order $\eps ^2$ per cell. Evaluating the energy of the trial state in this way we obtain an upper bound of the form 
\begin{equation}\label{eq:up bound Riemann}
 \anne \leq \sum_{n=1} ^{\neps} |s_{n+1} - s_n|\eone_\star (k_n) (1 + o(1))+ \OO (\eps ^2) 
\end{equation}
where the $o(1)$ error is due to the necessary modifications to~\eqref{eq:rough trial state} to make it continuous. The crucial point is to be able to control this error by showing that the modification needs not be a large one. This requires a detailed analysis of the $k$ dependence of the relevant quantities $\eonek$, $\alk$ and $\fk$ obtained by minimizing~\eqref{eq:intro 1D func disc}. Indeed, we prove in Section~\ref{sec:eff func} below that 
$$ \left| \eone_{\star} (k) -\eone_{\star} (k') \right| \leq C \eps\logi |k-k'|, \qquad  \left|\alpha (k) - \alpha (k') \right| \leq C \eps^{1/2} \logi |k-k'|^{1/2} $$
and, in a suitable norm,
$$   f_{k'} =  \fkstar + \OO \left(  \eps^{1/2} \logi |k-k'|^{1/2} \right),$$
which will allow to obtain the desired control of the $o(1)$ in~\eqref{eq:up bound Riemann} and conclude the proof by a Riemann sum argument.

\paragraph*{Lower bound.} In view of the argument we use for the upper bound, the natural idea to obtain the corresponding lower bound is to use the strategy for the disc case we developed in~\cite{CR2} locally in each cell. In the disc case, a classical method of energy decoupling  and Stokes' formula lead to the lower bound\footnote{We simplify the argument for pedagogical purposes.}  
\begin{equation}\label{eq:intro low bound disc}
\annf [\psi] \gtrapprox \eonek + \int_{\ann} \diff s \diff t \: \left(1 - \eps k t\right) \Kk (t) \left( |\dd_t v| ^2 + \tx\frac{\eps ^2}{\left(1 - \eps k t\right) ^2} |\dd_s v| ^2 \right) 
\end{equation}
where we have used the strict positivity of $\fk$ to write 
\begin{equation}\label{eq:intro decouple}
\psi (s,t)= \fk (t) e^{- i  \left( \eps ^{-1} \alk s -  \eps \deps s \right) } v(s,t) 
\end{equation}
and the ``cost function'' is 
\begin{align*}
 \Kk (t) &= \fk ^2 (t) + \Fk (t),\\
 \Fk(t) &= 2 \int_0 ^t d\eta \: \frac{\eta + \alk - \tx\frac{1}{2} \eps k \eta ^2}{1-\eps k \eta} \fk ^2 (\eta).
\end{align*}
This method is inspired from our previous works on the related Gross-Pitaevskii theory of rotating Bose-Einstein condensates~\cite{CRY,CPRY1,CPRY2} (informal summaries may be found in ~\cite{CPRY3,CPRY4}). Some of the steps leading to~\eqref{eq:intro low bound disc} have also been used before in this context~\cite{AH}. The desired lower bound in the disc case follows from~\eqref{eq:intro low bound disc} and the fact that $\Kk$ is essentially {\it positive}\footnote{More precisely it is positive except possibly for large $t$, a region that can be handled using the exponential decay of GL minimizers (Agmon estimates).} for any $k$. This is proved by carefully exploiting special properties of $\fk$ and~$\alk$.

To deal with the general case where the curvature is not constant, we again split the domain $\ann$ into small cells, approximate the curvature by a constant in each cell and use the above strategy locally.  A serious new difficulty however comes from the use of Stokes' formula in the derivation of~\eqref{eq:intro low bound disc}. We need to reduce the terms produced by Stokes' formula to expressions involving only first order derivatives of the order parameter, using further integration by parts. In the disc case, boundary terms associated with this operation vanish due to the periodicity of $\psi$ in the $s$ variable. When doing the integrations by parts in each cell, using different $\fk$ and $\alk$ in~\eqref{eq:intro decouple}, the boundary terms do not vanish since we artificially introduce some (small) discontinuity by choosing a cell-dependent profile $f_{k_n}$ as reference. 

To estimate these boundary terms we proceed as follows: the term at $s=s_{n+1}$, made of one part coming from the cell $ \cellj $ and one from the cell $\cell_{n+1}$ is integrated by parts back to become a bulk term in the cell $ \cellj $. In this sketch we ignore a rather large amount of technical complications and state what is essentially the conclusion of this procedure:   
\begin{equation}\label{eq:intro low bound gen}
\annf [\psi] \gtrapprox \sum_{n=1} ^{\neps} \bigg[ |s_{n+1} - s_n|\eone_\star (k_n) + \int_{\cellj} \diff s \diff t \lf(1 - \eps k_n t\right) \Kti \bigg( |\dd_t u_n| ^2 + \tx\frac{\eps ^2}{\left(1 - \eps k_n t\right) ^2} |\dd_s u_n| ^2 \bigg) \bigg]
\end{equation}
where 
\begin{equation}\label{eq:intro decouple bis}
u_n (s,t)= f_{k_n} ^{-1} (t) e^{ i  \left( \eps ^{-1} \alk s +  \eps \deps s \right) } \psi (s,t) 
\end{equation}
and the ``modified cost function'' is 
\begin{align*}
\Kti (s,t) &= K_{k_n} (t) - |\dd_s \chi_n (s)| |\ijj (t)| - |\chi_n(s)| \left| \dd_t \ijj (t)\right|,\\
\ijj (t) &= F_{k_n} (t)- F_{k_{n+1}} (t) \frac{f^2_{k_{n}}(t)}{f^2_{k_{n+1}}(t)}, 
\end{align*}
and $\chi_n$ is a suitable localization function supported in $\cellj$ with $\chi_n(s_{n+1}) = 1$ that we use to perform the integration by parts in $\cellj$. Note that the dependence of the new cost function on both $k_n$ and $k_{n+1}$ is due to the fact that the original boundary terms at $s_{n+1}$ that we transform into bulk terms in $ \cellj $ involved both $u_n$ and $u_{n+1}$.

The last step is to prove a bound of the form  
\begin{equation}\label{eq:intro bound L}
|\ijj(t)| + \left| \dd_t \ijj (t)\right| \leq C \eps \logi f_{k_n} ^2 (t)   
\end{equation}
on the ``correction function'' $\ijj$, so that 
$$ \Kti (t) \geq \left( 1- C \eps \logi \right)f_{k_n} ^2 (t) + F_{k_n} (t).$$
This allows us to conclude that (essentially) $\Kti \geq 0$ by a perturbation of the argument applied to $K_{k_n}$ in~\cite{CR2} and thus concludes the lower bound proof modulo the same Riemann sum argument as in the upper bound part. Note the important fact that the quantity in the l.h.s. of \eqref{eq:intro bound L} is proved to be small \emph{relatively to} $f_{k_n} ^2 (t) $, including in a region where the latter function is exponentially decaying. This bound requires a thorough analysis of auxiliary functions linked to~\eqref{eq:intro 1D func disc} and is in fact a rather strong manifestation of the continuity of this minimization problem as a function of~$k$. 

\medskip

The rest of the paper is organized as follows: Section~\ref{sec:functionals} contains the detailed analysis of the effective, curvature-dependent, 1D problem. The necessary continuity properties as function of the curvature are given in Subsection~\ref{sec:eff func} and the analysis of the associated auxiliary functions in Subsection~\ref{sec:aux func}. The details of the energy upper bound are then presented in Section~\ref{sec:up bound} and the energy lower bound is proved in Section~\ref{sec:low bound}. We deduce our other main results in Section~\ref{sec:density degree}. Appendix~\ref{sec:app} recalls for the convenience of the reader some material from~\cite{CR2} that we use throughout the paper.  

\section{Effective Problems and Auxiliary Functions}\label{sec:functionals}

This section is devoted to the analysis of the 1D curvature-dependent reduced functionals whose minimization allows us to reconstruct the leading and sub-leading order of the full GL energy. We shall prove results in two directions:
\begin{itemize}
\item We carefully analyse the dependence of the 1D variational problems as a function of curvature in Subsection~\ref{sec:eff func}. Our analysis, in particular the estimate of the subleading order of the GL energy, requires some quantitative control on the variations of the optimal 1D energy, phase and density when the curvature parameter is varied, that is when we move along the boundary layer of the original sample along the transverse direction. 
\item In our previous paper~\cite{CR2} we have proved the positivity property of the cost function which is the main ingredient in the proof of the energy lower bound in the case of a disc (constant curvature). As mentioned above, the study of general domains with smooth curvature that we perform here will require to estimate more auxiliary functions, which is the subject of Subsection~\ref{sec:aux func}.
\end{itemize}

We shall use as input some key properties of the 1D problem at fixed $k$ that we proved in~\cite{CR2}. These are recalled in Appendix~\ref{sec:app} below for the convenience of the reader.

\subsection{Effective 1D functionals}\label{sec:eff func}

We take for granted the three crucial but standard steps of reduction to the boundary layer, replacement of the vector potential and mapping to boundary coordinates. Our considerations thus start from the following reduced GL functional giving the original energy in units of $\eps ^{-1}$, up to negligible remainders: 
\begin{multline}\label{eq:GL func bound} 
	\annf[\psi] : = \int_0^{|\partial \Omega|} \diff s \int_0^{c_0 |\log\eps|} \diff t \lf(1 - \eps \curv t \ri) \lf\{ \lf| \partial_t \psi \ri|^2 + \frac{1}{(1- \eps \curv t)^2} \lf| \lf( \eps \partial_s + i \aae(s,t) \ri) \psi \ri|^2 \ri.	\\	
	\lf. - \frac{1}{2 \hex} \lf[ 2|\psi|^2 - |\psi|^4 \ri]  \ri\},
\end{multline}
where $k(s)$ is the curvature of the original domain. We have set
\beq\label{eq:vect pot bound}
	\aae(s,t) : =- t + \half \eps \curv t^2 + \eps \deps , 	
\eeq
and
\beq\label{eq:deps}
	\deps : = \frac{\gamma_0}{\eps^2} - \lf\lfloor \frac{\gamma_0}{\eps^2} \ri\rfloor,	\qquad	 \gamma_0 : = \frac{1}{|\partial \Omega|} \int_{\Omega} \diff \rv \: \curl \, \aavm,
\eeq
$ \lf\lfloor \: \cdot \: \ri\rfloor $ standing for the integer part. The boundary layer in rescaled coordinates is denoted by
\beq
	\label{eq:boundary layer}
	\ann : =  \lf\{ \rv \in \Omega \: | \: \dist(\rv, \partial \Omega) \leq c_0 \eps |\log\eps| \ri\}.
\eeq

The effective functionals that we shall be concerned with in this section are obtained by computing the energy \eqref{eq:GL func bound} of certain special states. In particular we have to go beyond the simple ans\"atze considered so far in the literature, e.g., in~\cite{FH-book,CR2}, and obtain the following effective energies:

\begin{itemize}
\item {\it 2D functional with definite phase}. Inserting the ansatz
\begin{equation}\label{eq:ansatz 2D}
\psi (s,t) = g(s,t) e^{- i  \left( \eps ^{-1} S (s) -  \eps \deps s \right) }  
\end{equation}
in~\eqref{eq:GL func bound}, with $g$ and $S$ respectively real valued density and phase, we obtain
\bml{\label{eq:2D func}
	\E^{\mathrm{2D}}_S[g] : = \int_{0}^{c_0 |\log\eps|} \diff t \int_0^{|\partial \Omega|} \diff s  \lf( 1 - \eps \curv t \ri) \lf\{ \lf| \partial_t g \ri|^2 + \frac{\eps^2}{(1 -  \eps \curv t)^2} \lf| \partial_s g \ri|^2 \ri.	\\
	+ \frac{\left(t + \dd_s S - \frac{1}{2}\eps t ^2 k (s) \right) ^2}{(1-\eps t k(s)) ^2} g^2 - \frac{1}{2b} \lf(2 g^2 - g^4 \ri) \bigg\}.
}
In the particular case where $\dd_s S = \alpha \in 2\pi \Z$ we may obtain a simpler functional of the density alone 
\bml{\label{eq:2D func bis}
	\E^{\mathrm{2D}}_\alpha [g] : = \int_{0}^{c_0 |\log\eps|} \diff t \int_0^{|\partial \Omega|} \diff s  \lf( 1 - \eps \curv t \ri) \lf\{ \lf| \partial_t g \ri|^2 + \frac{\eps^2}{(1 -  \eps \curv t)^2} \lf| \partial_s g \ri|^2 \ri.	\\
	\lf. + \pots(s,t) g^2 - \frac{1}{2b} \lf(2 g^2 - g^4 \ri) \ri\},
}
where
\begin{equation}\label{eq:pots}
\pots (s,t) =  \frac{\left(t + \alpha - \frac{1}{2} k (s) \eps t ^2 \right) ^2}{(1- k(s)\eps t) ^2}.
\end{equation}

However to capture the next to leading order of~\eqref{eq:GL func bound} we do consider a non-constant $\dd_s S$ to accommodate curvature variations, which is in some sense the main novelty of the present paper. In particular,~\eqref{eq:2D func bis} does \emph{not} provide the $ \OO(\eps)$ correction to the full GL energy. On the opposite \eqref{eq:2D func} does, once minimized over the phase factor $S$ as well as the density $g$. We will not prove this directly although it follows rather easily from our analysis. 
\item {\it 1D functional with given curvature and phase}. If the curvature $k(s)\equiv k$ is constant (the disc case), the minimization of~\eqref{eq:2D func bis} reduces to the 1D problem
\begin{equation}
\label{eq:1D func}
\fone_{k,\alpha}[f] : = \int_0^{c_0|\log\eps|} \diff t (1-\eps k t )\lf\{ \lf| \partial_t f \ri|^2 + \pot(t) f^2 - \tx\frac{1}{2b} \lf(2 f^2 - f^4 \ri) \ri\},
\end{equation}
with
\beq
	\label{eq:pot}
	\pot(t) : = \frac{(t + \alpha - \frac12 \eps k t ^2 )^2}{(1-\eps k t ) ^2}.
\eeq 
In the sequel we shall denote
\begin{equation}\label{eq:def interval}
\ieps = \left[0,c_0 |\log \eps|\right] = : [0,\teps]. 
\end{equation}
Note that~\eqref{eq:1D func} includes $ \OO(\eps)$ corrections due to curvature. As explained above our approach is to approximate the curvature of the domain as a piecewise constant function and hence an important ingredient is to study the above 1D  problem for different values of $k$, and prove some continuity properties when $k$ is varied. For $k=0$ (the half-plane case, sometimes referred to as the half-cylinder case) we recover the familiar 
\beq\label{eq:1D func bis}
\fone_{0,\alpha}[f] : = \int_0^{c_0|\log\eps|} \diff t \lf\{ \lf| \partial_t f \ri|^2 + (t + \alpha )^2 f^2 - \tx\frac{1}{2b} \lf(2 f^2 - f^4 \ri) \ri\},
\eeq
which has been known to play a crucial role in surface superconductivity physics for a long time (see~\cite[Chapter 14]{FH-book} and references therein). 
\end{itemize}

In this section we provide details about the minimization of~\eqref{eq:1D func} that go beyond our previous study~\cite[Section 3.1]{CR2}. We will use the following notation:
\begin{itemize}
\item Minimizing~\eqref{eq:1D func} with respect to $f$ at fixed $\alpha$ we get a minimizer $\fkal$ and an energy $\eone (k,\alpha)$.
\item Minimizing the latter with respect to $\alpha$ we get some $\alpha(k)$ and some energy $\eone_\star (k)$. {It follows from~\eqref{eq:vari 1D phase} below that $\alpha(k)$ is uniquely defined}.
\item Corresponding to $\eone_\star (k) : = \eone (k,\alpha(k)) $ we have an optimal density $\fkstar$, which minimizes $\eone (k,\alpha(k)) $, and a potential 
$$ \potk(t) : = V_{k,\alpha(k)}(t).$$
\end{itemize}

The following Proposition contains the crucial continuity properties (as a function of $k$) of these objects:

\begin{pro}[\textbf{Dependence on curvature of the 1D minimization problem}]\label{pro:1D curv}\mbox{}\\
Let $k,k' \in \R$ be bounded independently of $\eps$ and $1<b<\theo ^{-1}$, then the following holds: 
\begin{equation}\label{eq:vari 1D energy}
\left| \eone_{\star} (k) -\eone_{\star} (k') \right| \leq C \eps |k-k'| \logi
\end{equation}
and
\begin{equation}\label{eq:vari 1D phase}
|\alpha (k) - \alpha (k')| \leq C \left( \eps |k-k'| \right) ^{1/2} \logi. 
\end{equation} 
Finally, for all $n\in \N$,
\begin{equation}\label{eq:vari 1D opt density}
\left\Vert \fkstar ^{(n)}- f_{k'} ^{(n)} \right\Vert_{L^{\infty} (\ieps)} \leq C \left( \eps |k-k'| \right) ^{1/2} \logi. 
\end{equation}
\end{pro}

We first prove~\eqref{eq:vari 1D energy} and~\eqref{eq:vari 1D phase} and explain that these estimates imply the following lemma: 

\begin{lem}[\textbf{Preliminary estimate on density variations}]\label{lem:1D curv pre}\mbox{}\\
Under the assumptions of Proposition~\ref{pro:1D curv} it holds
\begin{equation}\label{eq:vari 1D opt density L 2}
\left\Vert \fkstar ^2 - f_{k'} ^2 \right\Vert_{L ^2 (\ieps)} \leq C \left( \eps |k-k'| \right) ^{1/2} \logi. 
\end{equation}
\end{lem}
 

\begin{proof}[Proof of Lemma~\ref{lem:1D curv pre}] We proceed in three steps:

\paragraph*{Step 1. Energy decoupling.}  We use the strict positivity of $\fk$ recalled in the appendix to write any function $f$ on $\ieps$ as 
$$ f = \fk v.$$
We can then use the variational equation \eqref{eq:var eq fal} satisfied by $ \fk $ to decouple the $\alpha ', k'$ functional in the usual way, originating in~\cite{LM}. Namely, we integrate by parts and use the fact that $\fk$ satisfies Neumann boundary conditions to write 
\begin{multline*}
\int_{0} ^{c_0 |\log \eps|}  \diff t \: (1-\eps k' t ) (\dd_t f)^2 = \int_{0} ^{c_0 |\log \eps|}  \diff t \: (1-\eps k' t )  \left[v^2 (\dd_t \fk) ^2 + \fk^2 (\dd_t v) ^2  + 2 \fk \dd_t  \fk v \dd_t v \right]\\
= \int_{0} ^{c_0 |\log \eps|} \diff t \: (1-\eps k' t)  \left[ \fk^2 (\dd_t v) ^2  + \lf( \tx\frac{\eps k'}{1-\eps k' t} - \tx\frac{\eps k}{1-\eps k t} \ri)  v^2 \fk \dd_t \fk  - \fk ^2 v^2 \left(\potk + \tx\frac{1}{b} (\fk ^2 - 1 )\right)\right].
\end{multline*}
Inserting this into the definition of $\fone_{k',\alpha '}$ and using~\eqref{eq:eone explicit}, we obtain for any $f$
\begin{align}\label{eq:decouple 1D}
\fone_{k',\alpha'} [f] &= \eonek  +  \fred [v] \nonumber
\\ &+ \int_{0} ^{c_0 |\log \eps|} \diff t \: (1-\eps k' t) \left(V_{k',\alpha'}(t) - \potk(t) \right) \fk ^2 v ^2 \nonumber
\\ &+  \frac{1}{b} \eps(k' -k )  \disp\int_{0} ^{c_0 |\log \eps|} \diff t \: t \fk ^4  + \eps \int_{0} ^{c_0 |\log \eps|} \diff t \: \left( k' - k \tx\frac{1-\eps k ' t}{1-\eps k t}\right) |v| ^2 \fk \dd_t \fk
\end{align}
with
\begin{equation}\label{eq:f red 1D}
\fred [v] = \int_{0} ^{c_0 |\log \eps|} \diff t \: (1-\eps k' t) \left\{ \fk  ^2 (\dd_t v ) ^2 + \tx\frac{1}{ 2b} \fk ^4 \left( 1 - v^2\right) ^2 \right\}. 
\end{equation}
In the case $\alpha'=\alpha(k)$ we can insert the trial state $v\equiv 1$ in the above, which gives 
\begin{equation}\label{eq:up bound ener 1D}
\eone_\star (k') \leq \eone_{k',\alpha(k)} \leq \eonek + C \eps |k-k'||\log \eps| ^{\infty}
\end{equation}
in view of the bounds on $\fk$ recalled in Appendix~\ref{sec:app} and the easy estimate
$$ \left| V_{k',\alk} (t)- \potk (t)\right| \leq C \eps |k-k'||\log \eps| ^{\infty}$$
for any $t\in \ieps$. Changing the role of $k$ and $k'$ in \eqref{eq:up bound ener 1D} we obtain the reverse inequality
$$\eonek \leq \eone_\star (k') + C \eps |k-k'||\log \eps| ^{\infty}, $$
and hence \eqref{eq:vari 1D energy} is proved.


\paragraph*{Step 2. Use of the cost function.} 
We now consider the case $\alpha' = \alpha (k'), f=f_{k'}$ and bound from below the term on the second line of~\eqref{eq:decouple 1D}. A simple computation gives  
\begin{multline}\label{eq:1D diff pot}
\int_{0} ^{c_0 |\log \eps|} \diff t \: (1-\eps k' t) \left(V_{k',\alpha(k')} - V_{k,\alpha(k)} \right) \fk ^2 v ^2 \\
= \int_{0} ^{c_0 |\log \eps|} \diff t \: (1-\eps k t) ^{-1} \left( \alpha(k')  - \alpha(k) \right) \left( 2t + \alpha (k) +  \alpha(k') - \eps k t ^2 \right) \fk ^2 v ^2 + \OO(\eps |k-k'|)
\\= \left( \alpha(k')  - \alpha(k) \right) ^2 \int_{0} ^{c_0 |\log \eps|} \diff t \: (1-\eps k' t) ^{-1} \fk ^2 v^2 
\\ + 2 ( \alpha(k') - \alpha(k)) \int_{0} ^{c_0 |\log \eps|} \diff t \: \frac{t + \alk  - \frac{1}{2}\eps k t ^2 }{1-\eps k t} \fk ^2 v ^2 + \OO(\eps |k-k'|).
\end{multline}
We may now follow closely the procedure of~\cite[Section 5.2]{CR2}: with the potential function $\Fk$ defined in~\eqref{eq:Fk} below we have 
$$2\frac{t + \alk  - \frac{1}{2}\eps k t ^2 }{1-\eps k t} \fk ^2 = \dd_t F_k (t)$$
and hence an integration by parts yields (boundary terms vanish thanks to Lemma~\ref{lem:F prop})
\begin{equation}\label{eq:1D mom term}
2 \int_{0} ^{c_0 |\log \eps|} \diff t \: \frac{t + \alk  - \frac{1}{2}\eps k t ^2 }{1-\eps k t} \fk ^2 v^2 = - 2 \int_{0} ^{c_0 |\log \eps|} \diff t \: \Fk v\dd_t v.
\end{equation}
We now split the integral into one part running from $0$ to $\btik$ and a boundary part running from $\btik$ to $c_0 |\log \eps|$, where $\btik$ is defined in \eqref{eq:annb} and~\eqref{eq:annb bis} below. For the second part, it follows from the decay estimates of Lemma~\ref{lem:point est fal} 
that 
\begin{equation}\label{eq:1D bound region}
	\int_{\btik} ^{c_0 |\log \eps|} \diff t \:  \Fk v\dd_t v  = \OO(\eps ^{\infty}). 
\end{equation}
To see this, one can simply adapt the procedure in~\cite[Eqs. (5.21) -- (5.28)]{CR2}. The bound~\eqref{eq:1D bound region} is in fact easier to derive than the corresponding estimate in~\cite{CR2} because the decay estimates in Lemma~\ref{lem:point est fal} are stronger than the Agmon estimates we had to use in that case. Details are thus omitted. 

We turn to the main part of the integral~\eqref{eq:1D mom term}, which lives in $[0,\btik].$ Since $\Fk$ is negative we have, using Lemma~\ref{lem:K positive} and Cauchy-Schwarz,
\begin{multline*}
\bigg| 2 (\alpha(k') - \alpha(k))  \int_{0} ^{\btik} \diff t \: \Fk v\dd_t v \bigg| \\ \leq 
 (\alpha(k') - \alpha(k)) ^2  \int_{0} ^{\btik} \diff t \: (1-\eps k' t) ^{-1}\left|\Fk\right| v^2 +  \int_{0} ^{\btik} \diff t \: (1-\eps k' t) \left|\Fk\right| (\dd_t  v )^2
\\\leq (1-\de) (\alpha(k') - \alpha(k)) ^2  \int_{0} ^{\btik} \diff t \: (1-\eps k' t) ^{-1} \fk ^2 v^2 + (1-\de) \int_{0} ^{\btik} \diff t (1-\eps k' t) \fk ^2 (\dd_t  v )^2
\end{multline*}
for any  $0< \de \leq C |\log \eps| ^{-4}$. Inserting this bound and~\eqref{eq:1D bound region} in~\eqref{eq:decouple 1D}, using~\eqref{eq:1D diff pot} and~\eqref{eq:1D mom term}, yields the lower bound 
\begin{align}\label{eq:1D low pre}
\eone_\star (k') &\geq \eonek  +  \int_{0} ^{c_0 |\log \eps|} \diff t \: (1-\eps t k') \left\{ \de \fk^2 (\dd_t v )^2 + \de  \frac{(\alpha' - \alpha(k)) ^2}{(1-\eps t k') ^2} \fk ^2 v^2 + \frac{\fk ^4}{ 2b} \left( 1 - v^2\right) ^2 \right\} \nonumber
\\ &+ \eps \int_{0} ^{c_0 |\log \eps|} \diff t  \: v^2 \fk \dd_t \fk \left( k' - k\frac{1-\eps tk '}{1-\eps t k}\right) - C \eps |k-k'| |\log \eps| ^{\infty}
\end{align}
where $v = f_{k'}/\fk$ and we also used the uniform bound~\eqref{eq:fal estimate} to estimate the fourth term of the r.h.s. of~\eqref{eq:decouple 1D}. 


\paragraph*{Step 3. Conclusion.} We still have to bound the first term in the second line of~\eqref{eq:1D low pre}: 
\begin{multline*}
\eps \int_{0} ^{c_0 |\log \eps|} \diff t \: v^2 \fk \dd_t \fk \left( k' - k \frac{1-\eps k ' t}{1-\eps k t}\right) 
 = \frac{1}{2} \left[ v^2 \fk ^2  \left( \eps k' - \eps k \frac{1-\eps k' t }{1-\eps k t}\right)\right]_0 ^{c_0 |\log \eps|} 
\\+ \int_{0} ^{c_0 |\log \eps|} \diff t \: v^2 \fk ^2 \frac{\eps k (k'-k)}{(1-\eps k t)^2} - \int_{0} ^{c_0 |\log \eps|} \diff t \: v \dd_t v \fk ^2 \left( \eps k' - \eps k \frac{1-\eps k' t }{1-\eps k t}\right).
\end{multline*}
The first two terms are both $\OO (\eps |k-k'| |\log \eps| ^{\infty})$ thanks to~\eqref{eq:fal estimate} applied to $f_{k'} ^2 = \fk ^2 v^2$. For the third one we write  
\bmln{
\bigg|\int_{0} ^{c_0 |\log \eps|} \diff t \: v \dd_t v \fk ^2 \left( \eps k' - \eps k \frac{1-\eps k' t }{1-\eps k t}\right)\bigg| \leq  C \eps |k-k'| |\log \eps| ^{\infty} \int_{0} ^{c_0 |\log \eps|} \diff t \: v |\dd_t v| \fk ^2
\\ \leq  C \eps |k-k'| |\log \eps| ^{\infty} \bigg[ \int_{0} ^{c_0 |\log \eps|} \diff t \: \fk ^2 v^2 + \int_{0} ^{c_0 |\log \eps|} \diff t \: \fk ^2 (\dd_t v) ^2 \bigg].
}
Inserting this in~\eqref{eq:1D low pre}, using again~\eqref{eq:fal estimate} and dropping a positive term, we finally get
\begin{align}\label{eq:1D low fin}
\eone_\star (k') &\geq \eonek  + |\log\eps|^{-5}  (\alpha(k') - \alpha(k)) ^2 \int_{0} ^{c_0 |\log \eps|} \diff t \: (1-\eps k' t) f_{k'} ^2  \nonumber
\\&+ \frac{1}{ 2b} \int_{0} ^{c_0 |\log \eps|} \diff t \: (1-\eps k' t) \left(\fk ^2 - f_{k'}^2 \right) ^2  - C \eps |k-k'| |\log \eps| ^{\infty}
\end{align}
where we have chosen $\de = |\log \eps | ^{-5}$, which is compatible with the requirement $0 < \de \leq C|\log \eps| ^{-4}$. Combining with the estimate~\eqref{eq:vari 1D energy} that we proved in Step 1 concludes the proof of~\eqref{eq:vari 1D phase}. To get \eqref{eq:vari 1D opt density L 2} one has to use in addition \eqref{eq:fal point l u b}, which guarantees that under the assumptions $1<b<\theo ^{-1}$ 
$$ \lf\| f_{k'} \ri\|_{L^2(\ieps)} \geq C > 0 $$
for some constant $C$ independent of $\eps$.
\end{proof}

To conclude the proof of Proposition~\ref{pro:1D curv} there only remains to discuss~\eqref{eq:vari 1D opt density}. We shall upgrade the estimate~\eqref{eq:vari 1D opt density L 2} to better norms, taking advantage of the 1D nature of the problem and using a standard bootstrap argument. 

\begin{proof}[Proof of Proposition~\ref{pro:1D curv}]
We write $\fk = f_{k'} + (\fk-f_{k'})$ and expand the energy $\eonek = \fonek [\fk]$, using the variational equation~\eqref{eq:var eq fal} for $f_{k'}$:
\begin{align*}
\eonek &\geq \eone_\star (k') +  \int_{\ieps} \diff t (1-\eps k t)|\dd_t (\fk - f_{k'})| ^2 + \int_{\ieps} \diff t (1-\eps k t) \potk (\fk - f_{k'}) ^2 
\\&+ \int_{\ieps} \diff t (1-\eps k t) ( \potk- V_{k'}) f_{k'} ^2 + 2 \int_{\ieps} \diff t (1-\eps k t) f_{k'} (\fk-f_{k'}) (\potk -V_{k'}) 
\\&+\frac{1}{2b}  \int_{\ieps} \diff t (1-\eps k t) \left[ 6 f_{k'} ^2 (\fk-f_{k'}) ^2 + 4 f_{k'} (\fk - f_{k'}) ^3  + (\fk - f_{k'}) ^4 - 2 (\fk - f_{k'}) ^2 \right] 
\\&-C \eps |k-k'| |\log \eps| ^{\infty}
\end{align*}
where the $\OO (\eps |k-k'| |\log \eps| ^{\infty})$ is as before due to the replacement of the curvature $k\leftrightarrow k'$. Using the same procedure to expand $\eone_\star (k') = \fone_{k'} [f_{k'}]$ and combining the result with the above we obtain
\begin{align*}
\eonek &\geq \eonek + 2 \int_{\ieps} \diff t (1-\eps k t)|\dd_t (\fk - f_{k'})| ^2 + \int_{\ieps} \diff t (1-\eps k t) (\potk + V_{k'}) (\fk - f_{k'}) ^2 
\\&+ \int_{\ieps} \diff t (1-\eps k t) ( \potk- V_{k'}) (f_{k'} ^2 - \fk ^2) 
\\&+ 2 \int_{\ieps} \diff t (1-\eps k t) ( f_{k'} (\fk-f_{k'}) - \fk (f_{k'} - \fk)) (\potk -V_{k'}) 
\\&+\frac{1}{2b}  \int_{\ieps} \diff t (1-\eps k t) (\fk-f_{k'}) ^2 \left[ 4 f_{k'} ^2 + 4 \fk ^2 + 4 f_{k'} \fk - 4 \right] 
\\&- C\eps |k-k'| |\log \eps| ^{\infty}.
\end{align*}
Hence it holds
\begin{align}\label{eq:1D improve bounds}
C \eps |k-k'| |\log \eps| ^{\infty} &\geq 2 \int_{\ieps} \diff t (1-\eps k t)|\dd_t (\fk - f_{k'})| ^2 \nonumber
\\&+ \int_{\ieps} \diff t (1-\eps k t) ( \potk- V_{k'}) (\fk ^2 - f_{k'} ^2 ) \nonumber
\\&+ \int_{\ieps} \diff t (1-\eps k t) (\fk-f_{k'}) ^2 \left[ \potk + V_{k'} + \frac{2}{b} \left( f_{k'} ^2 + \fk ^2 + f_{k'} \fk - 2 \right) \right].
\end{align}
Next we note that thanks to~\eqref{eq:vari 1D phase} 
$$ \sup_{\ieps} \left| \potk- V_{k'}\right| \leq C \left(|\alk - \alpha(k')| + \eps |k-k'|\right) |\log \eps| ^{\infty} \leq C \left(\eps |k-k'|\right) ^{1/2} |\log \eps| ^{\infty}$$
as revealed by an easy computation starting from the expression~\eqref{eq:pot}. Thus, using~\eqref{eq:vari 1D opt density L 2} and the Cauchy-Schwartz inequality,
\begin{multline}\label{eq:1D improve diff pot diff f} \left| \int_{\ieps} \diff t (1-\eps k t) ( \potk- V_{k'}) (\fk ^2 - f_{k'} ^2 )\right|\leq \\C |\log \eps | ^{1/2} \sup_{\ieps} \left| \potk- V_{k'}\right| \left\Vert \fk ^2 - f_{k'} ^2 \right\Vert_{L^2 (\ieps)} \leq C \eps |k-k'| |\log \eps| ^{\infty}.
\end{multline}

For the term on the third line of~\eqref{eq:1D improve bounds} we notice that, using the growth of the potentials $\potk$ and $V_{k'}$ for large $t$, the integrand is positive in 
$$\iepst:=\left[c_1(\log |\log \eps|)^{1/2}, c_0 |\log \eps|\right]$$
for any constant $c_1$ and $\eps$ small enough. On the other hand, combining~\eqref{eq:vari 1D opt density L 2} and the pointwise lower bound in~\eqref{eq:fal point l u b} we have  
$$ \left\Vert \fk  - f_{k'} \right\Vert_{L^2 (\iepst)}\leq C \left(\eps |k-k'| \right) ^{1/2} |\log \eps| ^{\infty}.$$
Splitting the integral into two pieces we thus have 
$$\int_{\ieps} \diff t (1-\eps k t) (\fk-f_{k'}) ^2 \left[ \potk + V_{k'} + \tx\frac{2}{b} \left( f_{k'} ^2 + \fk ^2 + f_{k'} \fk - 2 \right) \right] \geq - C \eps |k-k'| |\log \eps| ^{\infty}.$$
Using this and~\eqref{eq:1D improve diff pot diff f} we deduce from~\eqref{eq:1D improve bounds} that  
\begin{equation}\label{eq:1D H1 bound}
 \int_{\ieps} \diff t (1-\eps k t)|\dd_t (\fk - f_{k'})| ^2 \leq C \eps |k-k'| |\log \eps| ^{\infty} 
\end{equation}
and combining with the previous $L ^2$ bound this gives 
$$ \left\Vert \fk  - f_{k'} \right\Vert_{H^1 (\iepst)}\leq C \left(\eps |k-k'| \right) ^{1/2} |\log \eps| ^{\infty}.$$
Since we work on a 1D interval, the Sobolev inequality implies 
\begin{equation}\label{eq:1D L inf bound pre} \left\Vert \fk  - f_{k'} \right\Vert_{L^{\infty} (\iepst)}\leq C \left(\eps |k-k'| \right) ^{1/2} |\log \eps| ^{\infty}.
\end{equation}
In particular 
$$
\left| \fk (c_1(\log |\log \eps|)^{1/2})  - f_{k'}(c_1(\log |\log \eps|)^{1/2}) \right|\leq C \left(\eps |k-k'| \right) ^{1/2} |\log \eps| ^{\infty}.
$$
Then, integrating the bound~\eqref{eq:1D H1 bound} from $c_1(\log |\log \eps|)^{1/2}$ to $c_0 |\log \eps|$ we can extend~\eqref{eq:1D L inf bound pre} to the whole interval $\ieps$:
$$ \left\Vert \fk  - f_{k'} \right\Vert_{L^{\infty} (\ieps)}\leq C \left(\eps |k-k'| \right) ^{1/2} |\log \eps| ^{\infty},$$
which is~\eqref{eq:vari 1D opt density} for $n=0$. The bounds on the derivatives follow by a standard bootstrap argument, inserting the $L ^{\infty}$ bound in the variational equations.
\end{proof}

\subsection{Estimates on auxiliary functions}\label{sec:aux func}

In this Section we collect some useful estimates of other quantities involving the 1D densities as well as the optimal phases. It turns out that we need an estimate of the $k$-dependence of $\dd_t \log (f_k)$, provided in the following

	\begin{pro}[\textbf{Estimate of logarithmic derivatives}]
		\label{pro:est log der}
		\mbox{}	\\
		Let $k,k' \in \R$ be bounded independently of $\eps$ and $1<b<\theo ^{-1}$, then the following holds:
		\beq
			\label{eq:est log der}
			\lf\| \frac{f^{\prime}_k}{f_k} - \frac{f^{\prime}_{k'}}{f_{k'}} \ri\|_{L^{\infty}(\ie)} \leq C \left( \eps |k-k'| \right) ^{1/2} \logi.
		\eeq
	\end{pro}

	\begin{proof}
		Let us denote for short
		\beq
			g(t) : = \frac{f^{\prime}_k (t)}{f_k (t)}- \frac{f^{\prime}_{k'} (t)}{f_{k'} (t)}.
		\eeq
		We first notice that the estimate is obviously true in the region where $ f_k \geq |\log\eps|^{-M} $ for any $ M > 0 $ finite, thanks to \eqref{eq:vari 1D opt density} and \eqref{eq:fal derivative}:
		\bmln{
			|g(t)| \leq \frac{\lf| f^{\prime}_k - f^{\prime}_{k'}\ri|}{f_k} + \frac{\lf|f^{\prime}_{k'} \ri| \lf| f_k - f_{k'} \ri|}{f_k f_{k'}} \leq |\log\eps|^M \lf| f^{\prime}_k - f^{\prime}_{k'}\ri| + |\log\eps|^{M+3} \lf| f_k - f_{k'} \ri|	\\
			 \leq C \left( \eps |k-k'| \right) ^{1/2} \logi.
		}
		Let $ t_* $ be the unique solution to $ f_k(t_*) = |\log\eps|^{-M} $ (uniqueness follows from the properties of $ f_k $ discussed in Proposition \ref{pro:min fone}). To complete the proof it thus suffices to prove the estimate in the region $ [t_*, c_0 |\log\eps|] $. Notice also that thanks to \eqref{eq:fal point l u b}, it must be that $ t_* \to \infty $ when $\eps \to 0$.

		At the boundary of the interval $ [t_*, \teps] $ (recall \eqref{eq:def interval}), one has
		\beq
			g(t_*) = \OO\lf(\left( \eps |k-k'| \right) ^{1/2} |\log\eps|^{M}\ri),		\qquad	g(\teps) = 0,
		\eeq
		because of Neumann boundary conditions. Hence if the supremum of $ |g| $ is reached at the boundary there is nothing to prove. Let us then assume that $ \sup_{t \in [t_*,\teps]} |g| = |g(t_0)| $, for some $ t_* < t_0 < \teps $, such that $ g'(t_0) = 0 $, i.e.,
		\beq
			\label{log der step 0}
			\frac{f^{\prime\prime}_k(t_0)}{f_k(t_0)} - \frac{f^{\prime\prime}_{k'}(t_0)}{f_{k'}(t_0)} + \frac{\lf(f^{\prime}_k(t_0)\ri)^2}{f_k^2(t_0)} - \frac{\lf(f^{\prime}_{k'}(t_0)\ri)^2}{f^2_{k'}(t_0)} = 0.
		\eeq
		Since $ f_k $ and $ f_{k'} $ are both decreasing in $ [t_*,\teps] $ (see again Proposition \ref{pro:min fone}) we also have
		\beq
 			\label{log der step 1}	
			 \frac{\lf(f^{\prime}_k(t_0)\ri)^2}{f_k^2(t_0)} - \frac{\lf(f^{\prime}_{k'}(t_0)\ri)^2}{f^2_{k'}(t_0)} = \lf[ \frac{\lf|f^{\prime}_k(t_0)\ri|}{f_k(t_0)} + \frac{\lf|f^{\prime}_{k'}(t_0)\ri|}{f_{k'}(t_0)} \ri] g(t_0).
		\eeq
		The variational equations satisfied by $ f_k$ and $ f_{k'} $ on the other hand imply 
		\bml{
 			\label{log der step 2}
 			\bigg| \frac{f^{\prime\prime}_k(t_0)}{f_k(t_0)} - \frac{f^{\prime\prime}_{k'}(t_0)}{f_{k'}(t_0)} \bigg| = \bigg| \frac{\eps k f^{\prime}_k(t_0)}{(1 - \eps k t) f_k(t_0)} - \frac{\eps k' f^{\prime}_{k'}(t_0)}{(1 - \eps k' t) f_{k'}(t_0)} + V_k(t_0) - V_{k'}(t_0)	\\
			 - \frac{1}{b} \lf( f_k^2(t_0) - f_{k'}^2(t_0) \ri) \bigg| \leq C \lf[\lf(\eps|k - k'|\ri)^{1/2} \logi + \eps |g(t_0)| \ri],
		}
		thanks to \eqref{eq:vari 1D phase} and \eqref{eq:vari 1D opt density}. For the first two terms the estimate \eqref{eq:fal derivative} has also been used for the derivatives $ f_k^{\prime} $ and $ f_{k'}^{\prime} $:
		\bmln{
			 \frac{\eps k f^{\prime}_k(t_0)}{(1 - \eps k t) f_k(t_0)} - \frac{\eps k' f^{\prime}_{k'}(t_0)}{(1 - \eps k' t) f_{k'}(t_0)} = \OO(\eps) g(t_0) + \frac{f^{\prime}_{k'}(t_0)}{f_{k'}(t_0)} \lf( \frac{\eps k}{1 - \eps k t} - \frac{\eps k'}{1 - \eps k' t} \ri)	\\
			 = \OO(\eps) g(t_0) + \OO(\eps|k - k'|).
		}
	
		Plugging \eqref{log der step 1} and \eqref{log der step 2} into \eqref{log der step 0}, we get the estimate
		\beq
			\lf[ \frac{\lf|f^{\prime}_k(t_0)\ri|}{f_k(t_0)} + \frac{\lf|f^{\prime}_{k'}(t_0)\ri|}{f_{k'}(t_0)} + \OO(\eps) \ri] g(t_0) = \OO\lf(\lf(\eps|k - k'|\ri)^{1/2}\logi\ri).
		\eeq
		Now if
		\bdm
			\frac{\lf|f^{\prime}_k(t_0)\ri|}{f_k(t_0)} + \frac{\lf|f^{\prime}_{k'}(t_0)\ri|}{f_{k'}(t_0)} \geq |\log\eps|^{-2},
		\edm
		the result follows immediately. Therefore we can assume that
		\beq
			\frac{\lf|f^{\prime}_k(t_0)\ri|}{f_k(t_0)} + \frac{\lf|f^{\prime}_{k'}(t_0)\ri|}{f_{k'}(t_0)} \leq |\log\eps|^{-2},
		\eeq
		but we claim that this also implies  
		\beq
			\label{log der step 3}
			\frac{\lf|f^{\prime}_k(t)\ri|}{f_k(t)} + \frac{\lf|f^{\prime}_{k'}(t)\ri|}{f_{k'}(t)} \leq |\log\eps|^{-2} \mbox{ for any } t \in [t_0,\teps]. 
		\eeq
		Indeed, setting 
		$$ h_k(t) : = - f^\prime_k(t)/f_k(t),$$
		a simple computation involving the variational equation \eqref{eq:var eq fal} yields
		\bdm
			h_k^{\prime}(t) = - \frac{\eps k f^{\prime}_k(t)}{(1 - \eps k t) f_k(t)} - V_k(t) + \frac{1}{b} \lf(1 - f_k^2(t) \ri) + h_k^2(t) = - V_k(t) + h_k^2(t) + \OO(1),
		\edm
		using~\eqref{eq:fal derivative} again. Hence $ h_k^{\prime}(t_0) < 0 $, since $ V_k(t_0) \gg 1 $, which follows from $ t_0 > t_* \gg 1 $, and therefore~\eqref{log der step 3} holds. An identical argument applies to $ h_{k'} $ and thus to the sum 
		$$ h_k + h_{k'} = : h.$$
		Finally, the explicit expression of $ g^{\prime}(t) $ in combination with~\eqref{log der step 3} gives for $ t \geq t_0 $
		\bml{
 			|g(t)| = \bigg| \int_{t}^{\teps} \diff \eta \: g^{\prime}(\eta) \bigg| \leq \int_{t}^{\teps} \diff \eta \: \lf[ \lf(h(\eta) + \OO(\eps) \ri) \lf| g(\eta) \ri| + \OO\lf(\lf(\eps|k - k'|\ri)^{1/2}\logi\ri) \ri]	\\
			\leq C |\log\eps|^{-1} \sup_{t \in [t_0,\teps]} |g(t)| + \OO\lf(\lf(\eps|k - k'|\ri)^{1/2}\logi\ri),
		}
		which implies the result.
	\end{proof}

	The above estimate is mainly useful in providing bounds on quantities of the form
	\beq
		\label{eq:def ijj first}
		\ikk(t) : = F_k(t)  - F_{k'}(t) \frac{f_k ^2 (t)}{f_{k'} ^2(t)},
	\eeq
	alluded to  in Subsection~\ref{sec:sketch}. As announced there, the main difficulty is that we need to show that $\ikk$ is small \emph{relatively to} $f_k ^2$, which is the content of the following corollary. We need the following notation
	\beq
			[0,\tk] : = \lf\{ t : f_k(t) \geq |\log\eps|^3 f_k(\teps) \ri\}.
		\eeq
		Note that the monotonicity for large $ t $ of $ f_k $ guarantees that the above set is indeed an interval and that
		\beq
			\label{eq:tkk}
			\tk = \teps + \OO(\log|\log\eps|).
		\eeq

	\begin{cor}[\textbf{Estimates on the correction function}]
		\label{cor:est log cost}
		\mbox{}	\\
		Under the assumptions of Proposition \ref{pro:est log der}, it holds
		\beq
			\label{eq:est log cost}
			\sup_{t \in [0,\teps]} \bigg| \frac{\ikk}{f_k^2} \bigg| \leq C \left( \eps |k-k'| \right) ^{1/2} \logi
		\eeq
		and, setting $ \bteps : = \min\lf\{ \tk, \tkk \ri\} $,
		\begin{equation}
 		\label{eq:sup est der ijj}
		\sup_{t \in [0,\bteps]} \bigg| \frac{\partial_t \ikk}{f_k^2} \bigg| \leq C \left( \eps |k-k'| \right) ^{1/2} \logi.
		\end{equation}

	\end{cor}

	\begin{proof}
		We write 
		$$
		\frac{\ikk(t)}{f_k ^2 (t)} = \frac{F_k (t)}{f_k ^2 (t)} - \frac{F_{k'} (t)}{f_{k'} ^2 (t)}
		$$
		Using the definition of the potential function \eqref{eq:Fk} and its properties \eqref{F prop}, we can rewrite
		\bml{
 			\label{F ratio step 1}
			\frac{F_k(t)}{f^2_k(t)} - \frac{F_{k'}(t)}{f^2_{k'}(t)} = - \int_t^{\teps} \diff \eta \bigg[ b_k(\eta) \frac{f_k^2(\eta)}{f_k^2(t)} - b_{k'}(\eta) \frac{f_{k'}^2(\eta)}{f_{k'}^2(t)} \bigg]	\\
			= \int_t^{\teps} \diff \eta \bigg[ b_k(\eta) \bigg(\frac{f_{k'}^2(\eta)}{f_{k'}^2(t)} - \frac{f_k^2(\eta)}{f_k^2(t)}  \bigg) + \lf(b_{k'}(\eta) -  b_{k}(\eta)\ri) \frac{f_{k'}^2(\eta)}{f_{k'}^2(t)} \bigg].
		}
	
		We first observe that for any $ \eta \geq t $
		\beq
			\label{F ratio step 2}
			\frac{f_{k'}(\eta)}{f_{k'}(t)} \leq C,
		\eeq
		as it easily follows by combining the monotonicity of $ f_k $ for $ t $ large with its strict positivity close to the origin (see Proposition \ref{pro:min fone} and Lemma \ref{lem:point est fal} for the details). Hence we can bound the last term on the r.h.s. of \eqref{F ratio step 1} as
		\beq
			\label{F ratio step 3}
			\bigg| \int_t^{\teps} \diff \eta \: \lf(b_{k'}(\eta) -  b_{k}(\eta)\ri) \frac{f_{k'}^2(\eta)}{f_{k'}^2(t)} \bigg| \leq C |\log\eps| \lf\| b_{k'}-  b_{k} \ri\|_{L^{\infty}(\ie)} = \OO \lf( \left( \eps |k-k'| \right) ^{1/2} \logi \ri),
		\eeq
		since by \eqref{eq:vari 1D phase}
		\bdm
			 b_{k'}(t) -  b_{k}(t) = \lf(1 + \OO(\eps) \ri) \lf( \OO(\eps|k - k'|t^2)  + \al(k) - \al(k') \ri) =   \OO \lf( \left( \eps |k-k'| \right) ^{1/2} \logi \ri).
		\edm

		For the first term on the r.h.s. of \eqref{F ratio step 1} we exploit the estimate
		\bdm
			\frac{f_{k'}(\eta)}{f_{k'}(t)} - \frac{f_k(\eta)}{f_k(t)} =  \OO \lf( \left( \eps |k-k'| \right) ^{1/2} \logi \ri),
		\edm
		which can be proven by writing
		\bdm
			\frac{f_k(\eta)}{f_k(t)} = \exp\lf\{ \int_t^{\eta} \diff \tau \: \frac{f_k^{\prime}(\tau)}{f_k(\tau)} \ri\},
		\edm	
		which implies
		\bml{
 			\label{F ratio step 4}
 			\lf| \frac{f_{k'}(\eta)}{f_{k'}(t)} - \frac{f_k(\eta)}{f_k(t)} \ri| =  \frac{f_{k'}(\eta)}{f_{k'}(t)} \lf| 1 - \exp\lf\{ \int_t^{\eta} \diff \tau \: \bigg[ \frac{f_k^{\prime}(\tau)}{f_k(\tau)} - \frac{f_{k'}^{\prime}(\tau)}{f_{k'}(\tau)} \bigg] \ri\} \ri| 	\\
			\leq C  \int_t^{\eta} \diff \tau \: \bigg| \frac{f_k^{\prime}(\tau)}{f_k(\tau)} - \frac{f_{k'}^{\prime}(\tau)}{f_{k'}(\tau)} \bigg| \exp\lf\{ \int_t^{\eta} \diff \tau \bigg| \frac{f_k^{\prime}(\tau)}{f_k(\tau)} - \frac{f_{k'}^{\prime}(\tau)}{f_{k'}(\tau)} \bigg| \ri\}
			 \leq  C \left( \eps |k-k'| \right) ^{1/2} \logi,
		}
		where we have used \eqref{F ratio step 2}, the estimate $ |1 - e^{\delta}| \leq |\delta| e^{|\delta|} $, $ \delta \in \R $, and \eqref{eq:est log der}. Putting together \eqref{F ratio step 1} with \eqref{F ratio step 3} and \eqref{F ratio step 4}, we conclude the proof of~\eqref{eq:est log cost}. 
		
		To obtain~\eqref{eq:est log der} we first note that since $F_k' (t) \leq 0$,  the positivity of $ K_{k} $ in $ [0, \tk] $ recalled in Lemma~\ref{lem:K positive} ensures that 
		$$ \left|\frac{F_{k} (t)}{f_k ^2 (t)} \right|\leq 1$$
		in $ [0, \tk] $. Then we may use~\eqref{eq:est log der} again to estimate 
		\begin{multline*}
		\sup_{t \in [0,\bteps]} \bigg| \frac{\partial_t \ikk}{f_k^2} \bigg| = \sup_{t \in [0,\bteps]} \bigg[ \lf| \lf(1 - \eps k t \ri) b_k - \lf(1 - \eps k' t)\ri) b_{k'} \ri| + 2 \bigg| \frac{F_{k'}}{f_{k'}^2} \bigg| \bigg| \frac{f_k^{\prime}}{ f_{k}} - \frac{f_{k'}^{\prime}}{f_{k'}} \bigg| \bigg]	\\
		\leq  C \left( \eps |k-k'| \right) ^{1/2} \logi,
	\end{multline*}
	and the proof is complete.
	\end{proof}

\section{Energy Upper Bound}\label{sec:up bound}

We now turn to the proof of the energy upper bound corresponding to~\eqref{eq:energy GL}, namely we prove the following: 

\begin{pro}[\textbf{Upper bound to the full GL energy}]\label{pro:up bound}\mbox{}\\
Let $1<b<\theo ^{-1}$ and $\eps$ be small enough. Then it holds 
\begin{equation}\label{eq:up bound GL}
\glee \leq \frac{1}{\eps} \int_{0} ^{|\dd \Om|} \diff s \: \eone_\star (k(s)) + C \eps |\log \eps| ^{\infty}
\end{equation}
where $s\mapsto k(s)$ is the curvature function of the boundary $\dd \Om$ as a function of the tangential coordinate.
\end{pro}

This result is proven as usual by evaluating the GL energy of a trial state having the expected physical features. As is well-known~\cite{FH-book}, such a trial state should be concentrated along the boundary of the sample, and the induced magnetic field should be chosen close to the applied one. Before entering the heart of the proof, we briefly explain how these considerations allow us to reduce to the proof of an upper bound to the reduced functional~\eqref{eq:GL func bound}. We define 
\begin{equation}\label{eq:GL func bound inf}
\anne := \inf\left\{ \annf [\psi], \psi (0,t) = \psi (|\dd \Om|,t) \right\},
\end{equation}
the infimum of the reduced functional under periodic boundary conditions in the tangential direction and prove

\begin{lem}[\textbf{Reduction to the boundary functional}]\label{lem:up bound}\mbox{}\\
Under the assumptions of Proposition~\ref{pro:up bound}, it holds 
\begin{equation}\label{eq:up GL-GL bound}
\glee \leq \frac{1}{\eps} \anne + C \eps ^{\infty}.
\end{equation}
\end{lem}

\begin{proof}
This is a standard reduction for which more details may be found in~\cite[Section 14.4.2]{FH-book} and references therein. See also~\cite[Sections 4.1 and 5.1]{CR2}. We provide a sketch of the proof for completeness.

We first pick the trial vector potential as 
$$ \mathbf{A}_{\rm trial} = \mathbf{F} $$
where $\mathbf{F}$ is the induced vector potential written in a gauge where $\mathrm{div} \, \mathbf{F} = 0$, namely the unique solution of 
$$ \begin{cases}
    \mathrm{div}  \,\mathbf{F} = 0,		& \mbox{ in } \Om,\\
     \curl   \,\mathbf{F} = 1,				& \mbox{ in } \Om, \\
    \mathbf{F}\cdot \nuv = 0, 			& \mbox{ on } \dd \Om.
   \end{cases}
$$
Next we introduce boundary coordinates as described in~\cite[Appendix F]{FH-book}: let 
$$\gav(\xi): \R \setminus (|\partial \Omega| \Z) \to \partial \Omega $$
be a counterclockwise parametrization of the boundary $ \partial \Omega $ such that $ |\gav^{\prime}(\xi)| = 1 $. The unit vector directed along the inward normal to the boundary at a point $ \gav(\xi) $ will be denoted by $ \nuv(\xi) $. The curvature $ k(\xi) $ is then defined through the identity 
$$ \gav^{\prime\prime}(\xi) = k(\xi) \nuv(\xi). $$
Our trial state will essentially live in the region
\beq
	\label{ann}
	 \annt : = \lf\{ \rv \in \Omega \: | \: \dist(\rv, \partial \Omega) \leq c_0 \eps |\log\eps| \ri\},
\eeq
and in such a region we can introduce tubular coordinates $ ( s,\eps t) $ (note the rescaling of the normal variable) such that, for any given $ \rv \in \annt $, $ \eps t = \dist(\rv, \partial \Omega) $, i.e.,
\beq
	\label{eq:tubular coordinates}
	\rv(s,\eps t) = \gav^{\prime}( s) +  \eps t \nuv( s),
\eeq
which can obviously be realized as a diffeomorphism for $\eps$ small enough. Hence the boundary layer becomes in the new coordinates $ (s,t) $ 
\begin{equation}\label{eq:def ann rescale}
\ann:= \left\{ (s,t) \in \left[0, |\partial \Omega| \right] \times \left[0,c_0 |\log\eps|\right] \right\}.
\end{equation}
We now pick a function $\psi (s,t)$ defined on $\ann$, satisfying periodic boundary conditions in the $s$ variable. Using a smooth cut-off function $\chi (t) $ with $\chi (t) \equiv 1$ for $t\in [0,c_0 |\log \eps|]$ and $\chi (t)$ exponentially decreasing for $t > c_0 |\log \eps|$, we associate to $\psi$ the GL trial state
$$ \Psi_{\rm trial} (\rv) := \psi(s,t) \chi ( t) \exp \lf\{ i \phi_{\rm trial}(s,t) \ri\},$$
where $ \phi_{\rm  trial} $ is a gauge phase (analogue of \eqref{eq: gauge phase}) depending on $ \mathbf{A}_{\rm trial} $, i.e.,
\bml{
		\label{eq: gauge phase trial}
		\phi_{\rm trial}(s,t) : = - \frac{1}{\eps} \int_{0}^{t} \diff \eta \: \nuv(s) \cdot \mathbf{A}_{\rm trial}(\rv(s, \eps \eta)) + \frac{1}{\eps^2} \int_{0}^s \diff \xi \: \gav^{\prime}(\xi) \cdot \mathbf{A}_{\rm trial}(\rv(\xi,0)) \\ 
		-  \lf( \frac{|\Omega|}{|\partial \Omega| \eps^2 } - \lf\lfloor \frac{|\Omega|}{|\partial \Omega| \eps^2 } \ri\rfloor \ri) s.			
}
Then, with the definition of $\annf$ as in~\eqref{eq:GL func bound}, a relatively straightforward computation gives 
$$ \gle \left[\Psi_{\rm trial}, \mathbf{A}_{\rm trial} \right] \leq \frac{1}{\eps} \annf [\psi] + C \eps ^{\infty},$$
and the desired result follows immediately. Note that this computation uses the gauge invariance of the GL functional, e.g., through~\cite[Lemma F.1.1]{FH-book}.
\end{proof}

The problem is now reduced to the construction of a proper trial state for $\annf$. To capture the $\OO(\eps)$ correction (which depends on curvature) to the leading order of the GL energy (which does not depend explicitly on curvature), we need a more elaborate function than has been considered so far. The construction is detailed in Subsection~\ref{sec:trial state} and the computation completing the proof of Proposition~\ref{pro:up bound} is given in Subsection~\ref{sec:trial ener}.

\subsection{The trial state in boundary coordinates}\label{sec:trial state}

We start by recalling the splitting of the domain $ \ann $ defined in \eqref{eq:boundary layer} into $ \neps \propto \eps ^{-1} $ rectangular cells $ \lf\{ \cellj \ri\}_{n=1 \ldots \neps}$ with boundaries $s_n,s_{n+1}$ in the $ s$-coordinate such that 
$$s_{n+1} - s_n = \spac \propto \eps, $$
so that 
$$ \neps = \frac{|\partial \Omega|}{\spac}.$$
We denote
\beq
	\label{eq: cell}
	\cellj = [s_n,s_{n+1}] \times [0,c_0 |\log \eps|],
\eeq
with the convention that $ s_1 = 0 $, for simplicity. We will approximate the curvature $k(s)$ inside each cell by its mean value and set
\beq
	\label{eq:mean curvature}
	k_n := \spac^{-1} \int_{s_n} ^{s_{n+1}} \diff s \, k(s).
\eeq
We also denote by
\beq
	\label{eq: opt phase cell}
	\alpha_n = \alpha(k_n)
\eeq
the optimal phase associated to $k_n$, obtained by minimizing $\eone (\alpha,k_n)$ with respect to $\alpha$ as in Section~\ref{sec:eff func}.

The assumption about the smoothness of the boundary guarantees that
\beq
	\label{eq:curv diff}
	\kj - k_{n+1} = \OO(\eps).	
\eeq
Indeed if we assume that $ \sup_{s \in [0,2\pi]} \lf| \dd_s k (s) \ri| \leq C < \infty $ (independent of $ \eps $), one gets
\bmln{
 	\spac^{-1} \bigg| \int_{\kj}^{k_{n+1}} \diff s \: k(s) - \int_{k_{n+1}}^{k_{n+2}} \diff s \: k(s) \bigg| = \spac^{-1} \bigg| \int_{k_n}^{k_{n+1}} \diff s \int_{s}^{k_{n+1}} \diff \eta \: \dd_\eta k (\eta) +  \int_{k_{n+1}}^{k_{n+2}} \diff s \int_{k_{n+1}}^s \diff \eta \: \dd_\eta k(\eta)  \bigg|	\\
	 \leq C \spac = \OO(\eps).
}
We can then apply Proposition \ref{pro:1D curv} to obtain
\beq
	\label{eq:phase diff}
	\al_n - \al_{n+1} = \OO(\eps|\log\eps|^{\infty}),
\eeq
\beq
	\label{eq:density diff}
	\lf\|f^{(m)}_n - f^{(m)}_{n+1}\ri\|_{L^{\infty}(\ie)} = \OO(\eps|\log\eps|^{\infty}),
\eeq	
for any finite $ m \in \N $.


Our trial state has the form
\begin{equation}\label{eq:trial state}
\psit (s,t) = g(s,t) \exp \lf\{ -i\left(\eps^{-1} S(s)  - \eps \deps s\right) \ri\} 
\end{equation}
where $\deps$ is the number~\eqref{eq:deps}. The density $g$ and phase factor $S$ are defined as follows:

\begin{itemize}
\item \underline{The density.} The modulus of our wave function is constructed to be essentially piecewise constant in the $s$-direction, with the form $\fkstari (t)$ in the cell $ \cellj$. The admissibility of the trial state requires that $g$ be continuous and we thus set:
\begin{equation}\label{eq:trial density}
g(s,t):=  \fkstari + \chi_n,
\end{equation}
where the function $\chi_n$ satisfies 
\beq
\chi_n(s,t) = 
\begin{cases}
0,							&	\mbox{at } s = s_{n},\\
\fkstariplus(t) - \fkstari(t),		&	\mbox{at } s = s_{n+1}, 
\end{cases}
\eeq	
the continuity at the $s_n$ boundary being ensured by $\chi_{n-1}$. A simple choice is given by
\begin{equation}\label{eq:choice chi}
\chi_n (s,t)= \left(\fkstariplus (t) - \fkstari (t) \right) \left( 1 - \frac{s-s_{n+1}}{s_{n} - s_{n+1}}\right). 
\end{equation}
Note that $|k_n-k_{n+1}| \leq C |s_n-s_{n+1}|\leq C\eps$ since the curvature is assumed to be a smooth function of $s$. Clearly, in view of  Proposition~\ref{pro:1D curv} we can impose the following bounds on $\chi_n$:
\begin{equation}\label{eq:control smoothing}
|\chi_n| \leq C \eps \logi, \qquad |\dd_t \chi_n| \leq C \eps \logi, \qquad |\dd_s \chi_n| \leq C \logi,
\end{equation}
so that $\chi_n$ is indeed only a small correction to the desired density $\fkstari$ in $\cellj$.

\item \underline{The phase.} The phase of the trial function is dictated by the refined ansatz \eqref{eq:intro GLm formal refined}: within the cell $ \cellj $ it must be approximately equal to $ \al_n $ and globally it must define an admissible phase factor, i.e., vary of a multiple of $ 2\pi $ after one loop. We then let 
$$S = S (s) = \Sloc(s) + \Sglob (s)$$
where $\Sloc$ varies locally (on the scale of a cell) and $\Sglob$ varies globally (on the scale of the full interval $[0,|\dd \Om|]$) and is chosen to enforce the periodicity on the boundary of the trial state. The term $\Sloc$ is the main one, and its $s$ derivative should be equal to $\alpha_n$ in each cell $\cellj$ in order that the evaluation of the energy be naturally connected to the 1D functional we studied before{, as explained in Section 3.1}. We define $\Sloc$ recursively by setting:
\beq\label{eq:trial phase 1}
\Sloc (s)=
\begin{cases}
 \alpha_1 s, &  \mbox{in } \cell_1,\\
\alpha_{n} (s-s_{n}) + \Sloc(s_{n}), &	 \mbox{in } \cellj, n\geq 2,
\end{cases}
\eeq
which in particular guarantees the continuity of $\Sloc$ on $[s_1,s_{\neps+1}[$. Moreover we easily compute (recall that $ s_1 = 0 $)
\beq
	\label{eq:Sloc boundary}
	\Sloc(s_n) = \sum_{m = 2}^{n-1} \al_m \lf(s_{m+1} - s_m \ri) + \al_1 s_2 = \int_0^{s_{n}} \diff s \: \al(s) + \OO(\eps\logi).
\eeq

The factor $\Sglob$ ensures that 
$$S(s_{\neps+1}) - S (s_{1}) = S(s_{\neps+1}) \in 2 \pi \eps \Z,$$
which is required for~\eqref{eq:trial state} to be periodic in the $s$-direction and hence to correspond to a single-valued wave function in the original variables. The conditions we impose on $\Sglob$ are thus
\begin{align}\label{eq:trial phase 2}
\Sglob (s_1) &= 0 
\\ \Sglob(s_{\neps+1}) &= 2\pi \eps \left(  \alpha_{\neps} \lf(s_{\neps+1} - s_{\neps} \ri) + \Sloc(s_{\neps})    - \left\lfloor \alpha_{\neps} \lf(s_{\neps+1} - s_{\neps}\ri) + \Sloc(s_{\neps}) \right\rfloor \right)	\nonumber
\end{align}
with $\lfloor\,.\,\rfloor$ standing for the integer value. Thanks to \eqref{eq:Sloc boundary}, we have
$$ \alpha_{\neps} \lf(s_{\neps+1} - s_{\neps}\ri) + \Sloc(s_{\neps})  = \OO (1)$$
and we can thus clearly impose that $\Sglob$ be regular and  
\begin{equation}\label{eq:control S2}
|\Sglob| \leq C \eps,	\qquad |\dd_s \Sglob| \leq C \eps. 
\end{equation}

\begin{rem}($s$-dependence of the trial state)
	\mbox{}	\\
The main novelty here is the fact that the density and phase of the trial state have (small) variations on the scale of the cells which are of size $\OO(\eps)$ in the $s$-variable. A noteworthy point is that the phase needs not have a $t$-dependence to evaluate the energy at the level of precision we require. Basically this is associated with the fact that the $t^2$ term in~\eqref{eq:vect pot bound} comes multiplied with an $\eps$ factor.  The main point that renders the computation of the energy doable is~\eqref{eq:control smoothing} and this is where the analysis of Subsection~\ref{sec:eff func} enters heavily.
\end{rem}

\end{itemize}

\subsection{The energy of the trial state}\label{sec:trial ener}

We may now complete the proof of Proposition~\ref{pro:up bound} by proving

\begin{lem}[\textbf{Upper bound for the boundary functional}]\label{lem:up bound bound}\mbox{}\\
With $\psit$ given by the preceding construction, it holds 
\begin{equation}\label{eq:up bound bound}
\annf [\psit] \leq \int_{0} ^{|\dd \Om|} \diff s \eone_\star (k(s)) + \OO (\eps ^2 |\log \eps| ^{\infty}). 
\end{equation}
\end{lem}

The upper bound~\eqref{eq:up bound GL} follows from Lemmas~\ref{lem:up bound} and~\ref{lem:up bound bound} since $\psit$ is periodic in the $s$-variable and hence an admissible trial state for $\anne$.

\begin{proof} As explained in Subsection~\ref{sec:eff func}, inserting~\eqref{eq:trial state} into~\eqref{eq:GL func bound} yields 
\begin{equation}\label{eq:ener trial}
\annf [\psit] = \E^{\mathrm{2D}}_S [g] 
\end{equation}
where $\E^{\mathrm{2D}}_S [g]$ is defined in~\eqref{eq:2D func}. For clarity we split the estimate of the r.h.s. of the above equation into several steps. We  use the shorter notation $f_n$ for $f_{k_n}$ when this generates no confusion.

\paragraph{Step 1. Approximating the curvature.} In view of the continuity of the trial function, the energy is the sum of the energies restricted to each cell. We approximate $k(s)$ by $k_n$ in $\cellj$ as announced, and note that since $k$ is regular we have $|k(s) - k(s_n) |\leq  C \eps$ in each cell, with a constant $C$ independent of $j$. We thus have
\begin{multline}\label{eq:estim 1}
\E^{\mathrm{2D}}_S [g] \leq \sum_{n = 1}^{\neps} \int_{\cellj}  \diff t \, \diff s  \: \lf( 1 - \eps k_n t \ri) \lf\{ \lf| \partial_t g \ri|^2 + \frac{\eps^2}{(1 -  \eps k_n t)^2} \lf| \partial_s g \ri|^2 \ri.\\
	\lf. + \frac{\left(t + \dd_s S - \frac{1}{2}\eps t ^2 k_n \right) ^2}{(1-\eps k_n t) ^2} g^2 - \frac{1}{2b} \lf(2 g^2 - g^4 \ri) \ri\} \left( 1 + \OO(\eps ^2)\right)
\end{multline}
since each $k$-dependent term comes multiplied with an $\eps$ factor.

\paragraph{Step 2. Approximating the phase.} In $\cellj$ we have 
$$\dd_s S = \alpha_n + \dd_s \Sglob = \alpha_n + \OO (\eps).$$
We can thus expand the potential term: 
\begin{multline}\label{eq:change phase}
\int_{\cellj}  \diff t \, \diff s \: \frac{\left(t + \dd_s S - \frac{1}{2}\eps t ^2 k_n\right) ^2}{1-\eps k_n t} g^2  = \int_{\cellj}  \diff t \, \diff s \: \frac{\left(t + \alpha_n - \frac{1}{2}\eps t ^2 k_n \right) ^2}{1-\eps k_n t} g^2 
\\ 
+ 2 \int_{\cellj}  \diff t \, \diff s   \: \dd_s \Sglob \frac{t + \alpha_n - \frac{1}{2}\eps t ^2 k_n}{1-\eps k_n t} g^2 
+ \int_{\cellj}  \diff t \, \diff s \: \frac{\lf(\dd_s \Sglob \ri)^2}{1-\eps k_n t} g^2 
\end{multline}
and obviously  
$$ \int_{\cellj}  \diff t \, \diff s \: \frac{\lf(\dd_s \Sglob \ri)^2}{1-\eps k_n t} g^2 \leq C \eps ^3 \logi, $$ 
because of~\eqref{eq:control S2} and the size of $ \cellj $ in the $s$ direction. Next we note that in $ \cellj $
$$ g ^2  = f_n ^2 + 2 f_n \chi_n + \chi_n  ^2$$ 
so that, using~\eqref{eq:FH nonlinear} and the fact that $\dd_s \Sglob$ only depends on $s$ we have 
$$ \int_{\cellj}  \diff t \, \diff s   \: \dd_s \Sglob \frac{t + \alpha_n - \frac{1}{2}\eps t ^2 k_n}{1-\eps k_n t} g^2  = \int_{\cellj}  \diff t \, \diff s   \: \dd_s \Sglob \frac{t + \alpha_n - \frac{1}{2}\eps t ^2 k_n}{1-\eps k_n t} \left( 2 f_n\chi_n + \chi_n  ^2 \right),$$ 
which is easily bounded by $C\eps ^3 \logi$ using~\eqref{eq:control smoothing},~\eqref{eq:control S2} and the fact that $|s_{n+1}-s_n|\leq C \eps$. All in all: 
\begin{equation}\label{eq:change phase 2}
\int_{\cellj}  \diff t \, \diff s \: \frac{\left(t + \dd_s S - \frac{1}{2}\eps t ^2 k_n \right) ^2}{1-\eps k_n t} g^2 = \int_{\cellj}  \diff t \, \diff s \: \frac{\left(t + \alpha_n - \frac{1}{2}\eps t ^2 k_n \right) ^2}{1-\eps k_n t} g ^2 + \OO (\eps ^3 \logi). 
\end{equation}

\paragraph{Step 3. The 1D functional inside each cell.} We now have to estimate an essentially 1D functional in each cell, closely related to~\eqref{eq:1D func}:
\begin{equation}\label{eq:change density}
\int_{\cellj}  \diff t \, \diff s \: \lf( 1 - \eps k_n t \ri) \bigg\{ \lf| \partial_t g \ri|^2 + \frac{\eps^2}{(1 -  \eps k_n t)^2} \lf| \partial_s g \ri|^2 + \frac{\left(t + \alpha_n - \frac{1}{2}\eps t ^2 k_n \right) ^2}{(1-\eps k_n t)^2} g^2 - \frac{1}{2b} \lf(2 g^2 - g^4 \ri) \bigg\}. 
\end{equation}
We may now expand $g$ according to~\eqref{eq:trial density} in the above expression and use the variational equation~\eqref{eq:var eq fal} to cancel the first order terms in $\chi_n$. This yields
\begin{multline}\label{eq:change density 2}
\int_{\cellj}  \diff s \diff t \: \lf( 1 - \eps k_n t \ri) \lf\{ \lf| \partial_t g \ri|^2 + \tx\frac{\eps^2}{(1 -  \eps k_n t)^2} \lf| \partial_s g \ri|^2 + V_{k_n}(t) g^2 - \frac{1}{2b} \lf(2 g^2 - g^4 \ri) \ri\}  =  \spac \eone_{\star} (k_n) 
\\  + \int_{\cellj}  \diff s \diff t \: \lf( 1 - \eps k_n t \ri) \lf\{ |\dd_t \chi_n| ^2 + \tx\frac{\eps ^2}{(1-\eps t k_n) ^2} |\dd_s \chi_n| ^2 + V_{k_n} \chi_n ^2 + \frac{1}{2b} \left( 6 \chi_n ^2 f_n^2 + 4 \chi_n ^3 f_n + \chi_n ^4 - 2 \chi_n ^2 \right)\ri\} 
\\ = \spac \eone_{\star} (k_n) + O(\eps ^3 \logi),
\end{multline}
where we only have to use~\eqref{eq:control smoothing} to obtain the final estimate.

\medskip

\textbf{Step 4, Riemann sum approximation.} Gathering all the above estimates we obtain 
\beq\label{eq:up bound pre final}
\E^{\mathrm{2D}}_S [g] \leq 
\spac \sum _{n=1}^{\neps} \eone_{\star} (k_n) \left( 1+ \OO (\eps ^2)\right) + \OO(\eps ^2 \logi)
\\ = \int_0^{|\partial \Omega|} \diff s \: \eone_{\star} (k(s)) + \OO(\eps ^2 \logi). 
\eeq
Indeed, ~\eqref{eq:vari 1D energy} implies that inside $\cellj$
\beq
	\label{eq:1D energy continuous} \left| \eone_\star (k_n) -  \eone_\star (k (s))\right| \leq C \eps \spac \logi \leq C \eps ^2 \logi.
\eeq
Recognizing a Riemann sum of $ \neps \propto \eps ^{-1} $ terms in~\eqref{eq:up bound pre final} and recalling that $\eone_{\star} (k_n)$ is of order~$ 1 $, irrespective of $n$, thus leads to \eqref{eq:up bound pre final}. Combining~\eqref{eq:ener trial} and~\eqref{eq:up bound pre final} we obtain~\eqref{eq:up bound bound} which concludes the proof of Lemma~\ref{lem:up bound bound} and hence that of Proposition~\ref{pro:up bound}, via Lemma \ref{lem:up bound}.
\end{proof}

\section{Energy Lower Bound}
\label{sec:low bound}

The main result proven in this section is the following

	\begin{pro}[\textbf{Energy lower bound}]
		\label{pro:energy lb}
		\mbox{}\\
		Let $\Om\subset \R ^2$ be any smooth simply connected domain. For any fixed $1<b<\theo ^{-1}$, in the limit $ \eps \to 0$, it holds
		\begin{equation}\label{eq:energy lb}
			\glee \geq \frac{1}{\eps} \int_0^{|\partial \Omega|} \diff s \: \eone_\star \left(k(s)\right)  - C \eps |\log\eps|^{\infty}. 
		\end{equation}
	\end{pro}

We first reduce the problem to the study of decoupled functionals in the boundary layer in Subsection~\ref{sec:low bound pre} and then provide  lower bounds to these in Subsection~\ref{sec:low bound main}, which contains the main new ideas of our proof. 

\subsection{Preliminary reductions}\label{sec:low bound pre}
	
As in Section~\ref{sec:up bound}, the starting point is a restriction to the boundary layer together with a replacement of the vector potential. We refer to the proof of Lemma \ref{lem:up bound} and in particular \eqref{eq:tubular coordinates} for the definition of the boundary coordinates.

	\begin{lem}[\textbf{Reduction to the boundary functional}]
		\label{lem:low bound}
		\mbox{}\\
		Under the assumptions of Proposition \ref{pro:energy lb}, it holds
		\beq
			\label{eq:energy lb ann}
			\glee \geq \frac{1}{\eps} \annf[\psi] -C \eps^2 |\log\eps|^2,
		\eeq
		with $ \psi(s,t) = \glm(\rv(s,\eps t)) e^{-i \phi_{\eps}(s,t)} $ in $ \ann $, $ \phi_{\eps}(s,t) $ is a global phase defined in \eqref{eq: gauge phase} below and $\annf$ is the boundary functional defined in~\eqref{eq:GL func bound}
	\end{lem}

	\begin{proof}
		A simplified version of the result for disc samples is proven in \cite[Proposition 4.1]{CR2}, where a rougher lower bound is also derived for general domains. This latter result is obtained by dropping the curvature dependent terms from the energy, which was sufficient for the analysis contained there. Here we need more precision in order to obtain a remainder term of order $ o(\eps) $. We highlight here the main steps and skip most of the technical details.
	
		A suitable partition of unity together with the standard Agmon estimates (see~\cite[Section 14.4]{FH1}) allow to restrict the integration to the boundary layer:
		\beq
			\glee \geq \int_{\annt} \diff \rv \lf\{ \lf| \lf( \nabla + i \tx\frac{\aavm}{\eps^2} \ri) \Psi_1 \ri|^2 - \tx\frac{1}{2 \hex \eps^2} \lf[ 2|\Psi_1|^2 - |\Psi_1|^4 \ri]  \ri\} + \OO( \eps^{\infty}).
		\eeq
		where $ \Psi_1 $ is given in terms of $ \glm $ in the form $ \Psi_1 = f_1 \glm $ for some {function $ 0 \leq f_1 \leq 1 $, depending only on the normal coordinate $t$,} with support containing the set $ \annt $ defined by \eqref{ann} and contained in  
		\bdm
			\{ \rv \in \Omega \: | \: \dist(\rv, \partial \Omega) \leq C \eps|\log\eps| \}
		\edm
		 for a possibly large constant $ C $. The constant $ c_0 $ in the definition \eqref{ann} of the boundary layer has to be chosen large enough, but the choice of the support of $ f_1 $ remains to any other extent arbitrary and one can clearly pick $ f_1 $ in such a way that $ f_1 = 1 $ in $ \annt $ and going smoothly to $ 0 $ outside of it.

		The second ingredient of the proof is the replacement of the magnetic potential $ \aavm $ but this can be done following the same strategy applied to disc samples in \cite[Eqs. (4.18) -- (4.26)]{CR2}, whose estimates are not affected by the dependence of the curvature on $ s $. The crucial properties used there are indeed provided by the Agmon estimates, see below. The phase factor involved in the gauge transformation is explicitly given by
		\beq
			\label{eq: gauge phase}
			\phi_{\eps}(s,t) : = - \frac{1}{\eps} \int_{0}^{t} \diff \eta \: \nuv(s) \cdot \aavm(\rv(s, \eps \eta)) + \frac{1}{\eps^2} \int_{0}^s \diff \xi \: \gav^{\prime}(\xi) \cdot \aavm(\rv(\xi,0)) -  \deps s.			
		\eeq	

		The overall prefactor $ \eps^{-1} $ in the energy is then inherited from the rescaling of the normal coordinate $ \tau = \eps t $ in the tubular neighborhood of the boundary. Note here the use of a different convention with respect to both \cite{CR2,FH1}, where the tangential coordinate $ s $ was rescaled too.
	\end{proof}

	We need to rephrase some well-known decay estimates in a form suited to our needs. The Agmon estimates proven in \cite[Eq. (12.9)]{FH2} can be translated into analogous bounds applying to $ \psi(s,t) = \glm(\rv(s,\eps t))  e^{-i \phi_{\eps}(s,t)}  $ in $ \ann $: for some constant $A>0$ it holds
		{\beq
			\label{eq:Agmon est}
			\int_{\ann} \diff s \diff t \: (1 - \eps \curv t) \: e^{ A t } \lf\{ \lf| \psi(s,t) \ri|^2 + \lf| \lf( \lf(\eps\partial_s,\partial_t\ri) + i \tx\frac{\tilde{\aav}(s,t)}{\eps} \ri)  \psi(s,t) \ri|^2 \ri\} =  \OO(1),
 		\eeq
 		with (see, e.g., \cite[Eqs. (4.19) -- (4.20)]{CR2})
 		\beq
 			\label{eq: fake magnp}
 			\tilde{\aav}(s,t) : = \lf( (1 - \eps \curv t) \gav^{\prime}(s) \cdot \aavm(\rv(s,\eps t)) + \eps^2 \partial_s \phi_{\eps} \ri) {\bf e}_s.
		\eeq}
		In addition we are going to use two additional bounds proven in \cite[Eq. (10.21) and (11.50)]{FH2}:
		\beq
			\label{eq:sup est glm}
			\lf\| \psi \ri\|_{L^{\infty}(\ann)} \leq 1,	\qquad	\lf\| (\eps \partial_s, \partial_t) \psi \ri\|_{L^{\infty}(\ann)} \leq C.
		\eeq
	These bounds imply the following
	
	\begin{lem}[\textbf{Useful consequences of Agmon estimates}]
		\label{lem:Agmon est alternative}
		\mbox{}	\\
	 	Let $ \bar{t} = c_0|\log\eps| (1 + o(1)) $ for some $ c_0 $ large enough, then for any $ a, b, s_0 \in [0, 2\pi) $,
			\beq
				\label{eq:Agmon decay}
				\int_a^b \diff s \: |\psi(s,\bar{t})| = \OO(\eps^{\infty}),	\qquad \int_{\bar t}^{c_0|\log\eps|} \diff t \: |\psi(s_0,t)| = \OO(\eps^{\infty}).
			\eeq
	\end{lem}
	
	\begin{proof}
		We start by considering the first estimate: let $ \chi(t) $ be a suitable smooth function with support in $ [t_1,\bar{t}] $, with $ t_1 = \bar{t} - c $, for some $ c > 0 $, and such that $ 0 \leq \chi \leq 1$, $ \chi(\bar{t}) = 1 $ and $ |\partial_t \chi| \leq C $. Then one has
		\bml{
 			\int_a^b \diff s \: |\psi(s,\bar{t})| = \int_a^b \diff s \: \chi(\bar{t}) |\psi(s,\bar{t})| = \int_a^b \diff s \int_{t_1}^{\bar{t}} \diff t \: \lf[ \chi(t) \partial_t|\psi(s,t)| + |\psi(s,t)| \partial_t\chi(t) \ri] 	\\
			\leq C e^{-\frac{1}{2}A t_1} \bigg\{ \bigg[ \int_{\ann} \diff s \diff t \: e^{At} \lf| \partial_t|\psi(s,t)| \ri|^2 \bigg]^{1/2} + \bigg[ \int_{\ann} \diff s \diff t \: e^{At} \lf| \psi(s,t) \ri|^2 \bigg]^{1/2} \bigg\} = \OO(\eps^{\infty}),
		}
		by \eqref{eq:Agmon est}{, the diamagnetic inequality} and the assumption on $ t_1 $ and $ \bar{t} $. Indeed the factor $ e^{-\frac{1}{2}A t_1} = \eps^{\frac{1}{2}A c_0(1+o(1))} $ can be made smaller than any power of $ \eps $ by taking $ c_0 $ large enough. 
		
		For the second estimate we use a tangential cut-off function, i.e., a smooth monotone function $ \chi(s) $ with support\footnote{Let us assume that $ s_0 - 2\pi > C > 0 $, otherwise one can take as a support for $ \chi $ the complement set, i.e., $ [0, s_0] $.} in $ [s_0, 2\pi] $,  such that $ 0 \leq \chi \leq 1$, $ \chi(s_0) = 1 $, $ \chi(2\pi) = 0 $, and $ |\partial_s \chi| \leq C $. Then as in the estimate above (recall that $ \teps : = c_0|\log\eps| $)
		\bml{
 			\int_{\bar{t}}^{\teps} \diff t \: |\psi(s_0,t)| = \int_{\bar{t}}^{\teps} \diff t \: \chi(s_0) |\psi(s_0,t)| = - \int_{s_0}^{2\pi} \diff s \int_{\bar{t}}^{\teps} \diff t \: \lf[ \chi(s) \partial_s|\psi(s,t)| + |\psi(s,t)| \partial_s\chi(s) \ri] 	\\
			\leq C e^{-\frac{1}{2}A \bar{t}} \bigg\{ \eps^{-1} \bigg[ \int_{\ann} \diff s \diff t \: e^{At} \lf| \eps \partial_s|\psi(s,t)| \ri|^2 \bigg]^{1/2} + \bigg[ \int_{\ann} \diff s \diff t \: e^{At} \lf| \psi(s,t) \ri|^2 \bigg]^{1/2} \bigg\} = \OO(\eps^{\infty}),
		}
		where the main ingredients are again \eqref{eq:Agmon est}{, the diamagnetic inequality} and the assumption on~$\bar{t}$.
	\end{proof}	
	
	We now introduce some reduced energy functionals defined over the cells we have introduced before, see Subsection~\ref{sec:trial state} for the notation. 
		{We are going to perform an energy decoupling \`a la Lassoued-Mironescu~\cite{LM} in each cell:}  we write 
		\beq
			\label{eq:splitting psi}
			\psi(s,t) = : u_n(s,t) f_n(t) \exp \lf\{-i \tx\left(\frac{\alj}{\eps} + \deps\right)s \ri\},
		\eeq
		and introduce the reduced functionals
		\beq
			\label{eq:Ej}
			\E_n [u] : =  \int_{\cellj} \diff s \diff t \lf(1 - \eps \kj t \ri) f_n^2 \lf\{ \lf| \partial_t u \ri|^2 + \tx\frac{1}{(1- \eps \kj t)^2} \lf| \eps \partial_s u \ri|^2 - 2 \eps b_n(t) J_s[u] + \tx\frac{1}{2 \hex} f_n^2 \lf(1 - |u|^2 \ri)^2  \ri\},	
		\eeq
		with
		\beq
			b_n(t) : = \frac{t + \alj - \half \eps \kj t^2}{(1 - \eps \kj t)^2},
		\eeq
		and
		\beq
			J_s[u] : = (i u, \partial_s u) = \tx\frac{i}{2} \lf(u^* \partial_s u - u \partial_s u^*\ri).
		\eeq

	Note that in~\eqref{eq:Ej} the curvature is approximated by its mean value in the cell $\cellj$. These objects play a crucial role in the sequel, as per 
	\begin{lem}[\textbf{Lower bound in terms of the reduced functionals}]\label{lem:reduc func}\mbox{}\\
	 With the previous notation
	 \begin{equation}\label{eq:reduc func}
	  \annf [\psi] \geq \int_{0}^{|\partial \Omega|} \diff s \: \eones(k(s)) + \sum_{n=1} ^{\neps} \E_n[u_n] - C\eps^2 \logi
	 \end{equation}
	\end{lem}

	\begin{proof}
	With the above cell decomposition, we can estimate	
		\beq
			\label{step 0}
			\annf[\psi] \geq \sum_{n=1}^{\neps} \glfj[\psi] -C \eps ^2 |\log\eps|^{\infty},
		\eeq
		where
		\beq
			\glfj[\psi] : = \int_{\cellj} \diff s \diff t \lf(1 - \eps \kj t \ri) \lf\{ \lf| \partial_t \psi \ri|^2 + \tx\frac{1}{(1- \eps \kj t)^2} \lf| \lf(\eps \partial_s + i a_n(t) \ri) \psi \ri|^2  - \tx\frac{1}{2 \hex} \lf[ 2|\psi|^2 - |\psi|^4 \ri]  \ri\},
		\eeq
		and 
		\beq
			a_n(t) : = - t + \half \eps \kj t^2 + \eps \deps.
		\eeq
		The remainder term has been estimated as follows: the replacement of $ k(s) $ by $ k_n $ produces two different rests which can be estimated separately, i.e.,
		\beq
			\label{eq:lb remainder 1}
 			 \int_{\cellj} \diff s \diff t \: \lf(k(s) - \kj \ri) t \lf\{ \lf| \partial_t \psi \ri|^2 - \tx\frac{1}{2 \hex} \lf[ 2|\psi|^2 - |\psi|^4 \ri]  \ri\} = \OO(\eps^2|\log\eps|^{\infty}),
		\eeq
		\beq
			\label{eq:lb remainder 2}
			\frac{1}{\eps} \int_{\cellj} \diff s \diff t \: \lf\{ \tx\frac{1}{1- \eps k(s) t} \lf| \lf(\eps \partial_s + i \aae(s,t) \ri) \psi \ri|^2 -  \tx\frac{1}{1- \eps k_n t} \lf| \lf(\eps \partial_s + i a_n(t) \ri) \psi \ri|^2 \ri\} = \OO(\eps^2|\log\eps|^{\infty}).
		\eeq
		In estimating the first error term \eqref{eq:lb remainder 1}, we use the fact that 
		$$ k(s) - k_n = \OO(\eps)$$
		and the bounds \eqref{eq:sup est glm} together with the cell size. For the second estimate the same ingredients are sufficient as well, in addition to the simple bound
		\bdm
			\sup_{(s,t) \in \cellj} \lf| \aae(s,t) - a_n(t) \ri| \leq C \eps \sup_{(s,t) \in \cellj} \lf|k(s) - k_n \ri| |\log\eps|^2 = \OO(\eps^2 |\log\eps|^2).
		\edm

		Inside any given cell $ \cellj $ we can then decouple the functional in the usual way (see~\cite[Lemma~5.2]{CR2} for a statement in this context) to obtain
		\beq
			\label{eq:splitting}
			\glfj[\psi] = \eone_{\star}(k_n) \spac + \E_n[u_n].
		\eeq
		The first term in~\eqref{eq:splitting} is a Riemann sum approximation of the leading order term in~\eqref{eq:energy lb}: using \eqref{eq:1D energy continuous}, we immediately get
		\bml{
    			\label{eq:Riemann sum approx}
			\sum_{n = 1}^{\neps} \eone_{\star}(\kj) \spac = \sum_{n = 1}^{\neps} \eone_{\star}(\kj) (s_{n+1} - s_n)  	\\
			= \sum_{n = 1}^{\neps}\int_{s_n}^{s_{n+1}} \diff s \lf[ \eones(k(s)) + \OO(\eps^2 \logi) \ri] =  \int_{0}^{|\partial \Omega|} \diff s \: \eones(k(s)) + \OO(\eps^2 \logi),
		}
		which concludes the proof.
	\end{proof}

	\subsection{Lower bounds to reduced functionals}\label{sec:low bound main}
	
	In view of our previous reductions {in Lemma~\ref{lem:reduc func}}, the final lower bound~\eqref{eq:energy lb} is a consequence of the following lemma
	
	\begin{lem}[\textbf{Lower bound on the reduced functionals}]\label{lem:bound reduc func}\mbox{}\\
	 With the previous notation, we have 
	 \begin{multline}\label{eq:bound reduc func}
	 \sum_{n=1} ^{\neps} \E_n [u_n] \geq   |\log \eps| ^{-4} \sum_{n = 1}^{\neps} \int_{\cellj} \diff s \diff t \: \lf(1 - \eps \kj t \ri) f_n^2 \lf[ \lf| \partial_t u_n \ri|^2 + \tx\frac{1}{(1- \eps \kj t)^2} \lf| \eps \partial_s u_n \ri|^2 \ri] 	\\
		+ \disp\frac{1}{2 \hex \eps}\sum_{n = 1}^{\neps} \int_{\cellj} \diff s \diff t \: \lf(1 - \eps \kj t \ri)  f_n ^4 \lf( 1 - |u_n|^2 \ri)^2 - C \eps ^2 \logi
	 \end{multline}
	\end{lem}

	Proposition~\ref{pro:energy lb} now follows by a combination of Lemmas~\ref{lem:low bound},~\ref{lem:reduc func} and~\ref{lem:bound reduc func} because the two sums in the right-hand side of~\eqref{eq:bound reduc func} are positive. These terms will prove useful to obtain our density and degree estimates in Section~\ref{sec:density degree}.
	
	We can now focus on the proof of Lemma~\ref{lem:bound reduc func}, which is the core argument of the proof of Proposition~\ref{pro:energy lb}.

	\begin{proof}[Proof of Lemma \ref{lem:bound reduc func}] 
		
		The proof is split into two rather different steps. In the first one we essentially follow the strategy of~\cite[Section 5.2]{CR2} to control the main part of the only potentially negative term in~\eqref{eq:Ej}. This is done locally inside each cell and uses mainly the positivity of the cost function, Lemma~\ref{lem:K positive}. This strategy however involves an application of Stokes' formula and subsequent further integrations by parts to put the so obtained terms in such a form (involving only first order derivatives, see \eqref{step 1}) that they can be compared with the kinetic one. This produces unphysical surface terms located on the boundaries of the (rather artificial) cells we have introduced. The second step of the proof consists in controlling those, which requires to sum them all and reorganize the sum in a convenient manner. It is in this step only that we cease working locally inside each cell. 
		
		\paragraph{Step 1. Lower bound inside each cell.} First, we split the integration over two regions, one where a suitable lower bound to the density $ f_n $ holds true and another one yielding only a very small contribution. More precisely we set
		\beq
			\rest : = \lf\{ (s,t) \in \cellj: f_n(t) \geq |\log\eps|^3 f_n(\teps) \ri\}.
		\eeq
		Note that the monotonicity for large $ t $ of $ f_n $ (see Proposition \ref{pro:min fone}) guarantees that 
		\beq
			\label{tj}
			\rest : =  [s_n, s_{n+1}] \times [\tj, \teps],	\qquad	\tj = \teps + \OO(\log|\log\eps|).
		\eeq
		Now we use the potential function $ F_n(t) $ defined as
		\beq
			F_n(t) : = 2 \int_0^t \diff \eta \: (1 - \eps k_n \eta) f_n^2(\eta) b_n(\eta) =  2 \int_0^t \diff \eta \: f_n^2(\eta) \frac{\eta + \alj - \half \eps \kj \eta^2}{1 - \eps \kj \eta},
		\eeq
		and compute
		\begin{equation} 
			\label{eq:step 0}
 			- 2 \eps  \int_{\cellj} \diff s \diff t \: \lf(1 - \eps k_n t \ri) f_n^2(t) \beps(t) J_s[u_n] = \eps \int_{\cellj} \diff s \diff t \:  F_n (t) \partial_t J_s[u_n],
		\end{equation}
		where we have exploited the vanishing of $ F_n $ at $ t = 0 $ and $ t = \teps $. Now we split the r.h.s. of the above expression into an integral over $ \domj : = \cellj \setminus \rest $ and a rest. In order to compare the first part with the kinetic energy and show that the sum is positive, we have to perform another integration by parts:
		\bml{
 			\label{step 1}
			\eps \int_{\domj} \diff s \diff t \: F_n(t) \partial_t J_s[u_n] = 2 \eps \int_0^{\tj} \diff t \int_{s_n}^{s_{n+1}} \diff s \: F_n(t)  \lf( i  \partial_t u_n, \partial_s u_n \ri) \\ + \eps \int_0^{\tj} \diff t F_n(t) \lf[ J_t[u_n](s_{n+1}, t) -  J_t[u_n](s_{n}, t) \ri]. 
		}
		The first term in~\eqref{step 1} can be bounded by using some kinetic energy:
		\bml{
 			\label{step 3}
			 2 \eps \int_{\domj} \diff t \diff s \: F_n(t)  \lf( i  \partial_t u_n, \partial_s u_n \ri) \geq - 2 \int_{\domj} \diff s \diff t \: \lf| F_n (t) \ri| \lf|  \partial_t u_n \ri| \lf| \eps \partial_s u_n \ri| \\
			 \geq - \int_{\domj} \diff s \diff t \: (1 - \eps \kj t) F_n (t) \lf[ \lf| \partial_t u_n \ri|^2 + \tx\frac{1}{(1 - 
			 \eps \kj t)^2} \lf| \eps \partial_s u_n \ri|^2 \ri],
		}
		where we have used the inequality $ab\leq \frac{1}{2} (\delta a^2 + \delta ^{-1} b^2)$ and the negativity of $ F_n(t) $ (see Lemma~\ref{lem:F prop}). Combining the above lower bound with~\eqref{eq:Ej} and~\eqref{step 0} and dropping the part of the kinetic energy located in $ \rest $, we get
		\bml{
 			\label{step 4}
			\E_n[u_n] \geq  \int_{\domj} \diff s \diff t \: \lf(1 - \eps \kj t \ri) K_n(t) \bigg[ \lf| \partial_t u_n \ri|^2 + \tx\frac{1}{(1- \eps \kj t)^2} \lf| \eps \partial_s u_n \ri|^2 \bigg] 	\\
			+ \eps \int_0^{\tj} \diff t F_n(t) \lf[ J_t[u_n](s_{n+1}, t) -  J_t[u_n](s_{n}, t) \ri] + \eps \disp\int_{\rest} \diff s \diff t \:  F_n (t) \partial_t J_s[u_n]	\\
			 + \de \int_{\cellj} \diff s \diff t \: \lf(1 - \eps \kj t \ri) f_n^2 \lf[ \lf| \partial_t u_n \ri|^2 
			+ \tx\frac{1}{(1- \eps \kj t)^2} \lf| \eps \partial_s u_n \ri|^2 \ri] 	\\+  \disp\frac{1}{2 \hex} \int_{\cellj} \diff s \diff t \: \lf(1 - \eps \kj t \ri)  f_n  ^4 \lf( 1 - |u_n|^2 \ri)^2,
		}
		where 
	\beq
		\label{eq:Kj}
		K_n (t) : = K_{k_n}(t),	
	\eeq
	is the cost function defined in \eqref{eq:Kk}, for some given $\de$, satisfying \eqref{eq:de}. The third term in \eqref{step 4} is bounded from below by a quantity smaller than any power of $ \eps $, provided $ c_0 $ is chosen large enough. This is shown using the same strategy as in \cite[Eq. (5.21) and following discussion]{CR2} and we skip the details for the sake of brevity. For the first term we use the positivity of $K_n$ provided by Lemma~\ref{lem:K positive}. We then conclude 
	\bml{
 			\label{eq:step 4}
			\E_n[u_n] \geq  \eps \int_0^{\tj} \diff t F_n(t) \lf[ J_t[u_n](s_{n+1}, t) -  J_t[u_n](s_{n}, t) \ri] \\
			 + \de \int_{\cellj} \diff s \diff t \: \lf(1 - \eps \kj t \ri) f_n^2 \lf[ \lf| \partial_t u_n \ri|^2 
			+ \tx\frac{1}{(1- \eps \kj t)^2} \lf| \eps \partial_s u_n \ri|^2 \ri] 	\\+  \disp\frac{1}{2 \hex} \int_{\cellj} \diff s \diff t \: \lf(1 - \eps \kj t \ri)  f_n  ^4 \lf( 1 - |u_n|^2 \ri)^2 + \OO (\eps ^{\infty}),
		}
	and there only remains to bound the first term on the r.h.s. from below. We are not actually able to bound the term coming from cell $ n $ separately, so in the next step we put back the sum over cells.    
		
	\paragraph{Step 2. Summing and controlling boundary terms.} We now conclude the proof of~\eqref{eq:bound reduc func} by proving the following inequality: 
	\begin{multline}\label{eq:control boundary terms}
		\eps \sum_{n=1} ^{\neps} \int_0^{\tj} \diff t F_n(t) \lf[ J_t[u_n](s_{n+1}, t) -  J_t[u_n](s_{n}, t) \ri] \geq \\ -C|\log \eps| ^{-5} \sum_{n=1} ^{\neps} \int_{\cellj} \diff s \diff t \: \lf(1 - \eps \kj t \ri) f_n^2 \lf[ \lf| \partial_t u_n \ri|^2 
			+ \tx\frac{1}{(1- \eps \kj t)^2} \lf| \eps \partial_s u_n \ri|^2 \ri] -C \eps ^2 \logi.
	\end{multline}
	Grouping~\eqref{eq:step 4} and~\eqref{eq:control boundary terms}, choosing $\de = 2 |\log \eps| ^{-4}$ (which we are free to do) concludes the proof. 
	
	We turn to our claim~\eqref{eq:control boundary terms}. Once we have put back the sum over all cells the idea is to associate the two terms evaluated on the same boundary, which come from two adjacent cells and therefore contain two different densities:
	\bml{
		\label{step 5}
		\eps \sum_{n=1}^{\neps}  \int_0^{\tj} \diff t F_n(t) \lf[ J_t[u_n](s_{n+1}, t) -  J_t[u_n](s_{n}, t) \ri] 	\\
		= \eps \sum_{n=1}^{\neps} \bigg[ \int_0^{\tj} \diff t \lf[ F_n(t) J_t[u_n](s_n,t) - F_{n+1}(t) J_t[u_{n+1}](s_n,t) \ri] + R_n \bigg],
	}
	where, assuming without loss of generality that $ \tj < \tjj $, 
	\beq
		R_n : = - \int_{\tj}^{\tjj} \diff t \: F_{n+1}(t) J_t[u_{n+1}](s_{n+1},t),
	\eeq
	If on the other hand $  \tj > \tjj $, in~\eqref{step 5} $ \tj $ should be replaced with $ \tjj $ and in place of $ R_n $ one would find
	\bdm
		\int_{\tjj}^{\tj} \diff t \: F_{j}(t) J_t[u_{j}](s_{n+1},t).
	\edm 
	In other words the remainder $ R_n $ is inherited from the fact that the decomposition $ \cellj = \domj \cup \rest $ clearly depends on $ n $ and the boundary terms in~\eqref{step 5} do not compensate exactly. However it is clear from what follows that the estimate of such a boundary term is the same in both cases and essentially relies on the second inequality in~\eqref{eq:Agmon decay}: recalling that 
	$$ f_{n+1}^2(t) + F_{n+1}(t) \geq 0 $$
	for any $ t \leq \tjj $, we have
	\bml{
 		|R_n| = \int_{\tj}^{\tjj} \diff t \: \lf| F_{n+1}(t) \ri| \lf| J_t[u_{n+1}](s_{n+1},t) \ri| \leq  \int_{\tj}^{\tjj} \diff t \: f_{n+1}^2(t) \lf|u_{n+1}(s_{n+1},t) \ri| \lf| \partial_t u_{n+1}(s_{n+1},t) \ri| 	\\
		\leq \int_{\tj}^{\tjj} \diff t \: \lf|\psi(s_{n+1},t) \ri| \lf[ \lf| \partial_t \psi (s_{n+1},t) \ri| +  \lf|u_{n+1}(s_{n+1},t) \ri| \lf| \partial_t f_{n+1} (t) \ri| \ri] \\
		\leq C |\log\eps|^3 \int_{\tj}^{\tjj} \diff t \: \lf|\psi(s_{n+1},t) \ri| = \OO(\eps^{\infty}),
	}
	where we have used the bounds \eqref{eq:sup est glm} and \eqref{eq:fal derivative}, i.e., $ |f_{n+1}^{\prime}| \leq |\log\eps|^3 f_{n+1}(t) $. The identity \eqref{step 5} hence yields
	\bml{
		\label{step 6}
 		 \eps \sum_{n=1}^{\neps}  \int_0^{\tj} \diff t F_n(t) \lf[ J_t[u_n](s_{n+1}, t) -  J_t[u_n](s_{n}, t) \ri]  \\
		= \eps \sum_{n=1}^{\neps} \int_0^{\tj} \diff t \lf[ F_n(t) J_t[u_n](s_n,t) - F_{n+1}(t) J_t[u_{n+1}](s_n,t) \ri]  + \OO(\eps^{\infty}).
	}
	Using now the definitions~\eqref{eq:splitting psi} of $ u_n $ and $ u_{n+1} $, we get
	\beq
		u_{n+1}(s,t) = \frac{f_n(t)}{f_{n+1}(t)} e^{i(\aljj - \alj) s} u_n(s,t),
	\eeq
	so that
	\beq
		J_t[u_{n+1}](s_n,t) = i \gjj(t) \gjj^{\prime}(t) \lf| u_n(s_n,t) \ri|^2 + \gjj^2(t) J_t[u_n](s_n,t),
	\eeq
	where we have set
	\beq
		\gjj(t) : = \frac{f_n(t)}{f_{n+1}(t)}.
	\eeq
	Then we can compute
	\begin{multline*}
 		\eps \int_0^{\tj} \diff t \lf[ F_n J_t[u_n](s_n,t) - F_{n+1} J_t[u_{n+1}](s_n,t) \ri] 
 		\\= \eps\int_0^{\tj} \diff t \: \lf[ F_n(t)  - F_{n+1}(t) \gjj^2(t) \ri] J_t[u_n](s_n,t)	\\
		\\- \frac{i\eps}{2} \int_0^{\tj} \diff t \:  F_{n+1}(t) \partial_t \lf( \gjj^2(t) \ri) \lf| u_n(s_n,t) \ri|^2,	
	\end{multline*}
	but we know that the l.h.s. of the above expression is real, so that we can take the real part of the identity above obtaining
	\beq
		\eps \int_0^{\tj} \diff t \lf[ F_n J_t[u_n](s_n,t) - F_{n+1} J_t[u_{n+1}](s_n,t) \ri] = \eps \int_0^{\tj} \diff t \: \lf[ F_n  - F_{n+1} \gjj^2 \ri] J_t[u_n](s_n,t).
	\eeq
	To estimate the r.h.s. we integrate by parts back by introducing a suitable cut-off function. Let, for any given $ n = 1, \ldots, \neps $, $ \chi_n(s) $ be a suitable smooth function, such that 
	$$ \chi_n(s_n) = 1, \qquad \chi\lf(\tx\frac12 ( s_{n} + s_{n+1})\ri) = 0 $$
	and
	\beq
		 \lf[  s_n, \half \lf( s_{n} + s_{n+1} \ri) \ri)\subset \supp(\chi_n), 	\qquad	\lf| \partial_s \chi_n \ri| \leq C{\eps^{-1}}.
	\eeq
	We can rewrite
	\bml{
 		\label{step 7}
 		\eps \int_0^{\tj} \diff t \: \lf[ F_n  - F_{n+1} \gjj^2 \ri] J_t[u_n](s_n,t) = \eps \int_0^{\tj} \diff t \: \chi_n(s_n) \lf[ F_n  - F_{n+1}  \gjj^2 \ri] J_t[u_n](s_n,t)	\\
		 =  \eps \int_0^{\tj} \diff t \int_{s_n}^{\frac12 \lf( s_{n} + s_{n+1}\ri)} \diff s \Big\{ \chi_n(s) \ijj(t) \partial_s \lf( J_t[u_n] \ri) + \partial_s \lf( \chi_n(s) \ri) \ijj(t) J_t[u_n]  \Big\},
	}
	where we have set for short (compare with~\eqref{eq:def ijj first})
	\beq
		\label{eq:def ijj}
		\ijj(t) : = F_n(t)  - F_{n+1}(t) \gjj^2(t) = F_n(t)  - F_{n+1}(t) \frac{f_n ^2 (t)}{f_{n+1}^2(t)}.
	\eeq
	The first contribution to~\eqref{step 7} can be cast in a form analogous to~\eqref{step 3}:
	\bml{
 		\label{step 8}
 		\eps \int_0^{\tj} \diff t \int_{s_n}^{\frac12 \lf( s_{n} + s_{n+1}\ri)} \diff s \: \chi_n(s) \ijj(t) \partial_s \lf( J_t[u_n] \ri) \\
		= \eps \int_0^{\tj} \diff t \int_{s_n}^{\frac12 \lf( s_{n} + s_{n+1}\ri)} \diff s  \: \chi_n(s) \lf\{ 2 \ijj(t) \lf(i \partial_s u_n, \partial_t u_n \ri) -  \partial_t \lf(\ijj(t)\ri) J_s[u_n]  \ri\}	\\
		 + \eps \int_{s_n}^{\frac12 \lf( s_{n} + s_{n+1}\ri)} \diff s \: \chi_n(s) \ijj(\tj) J_s[u_n](s,\tj).
	}
	 The first term on the r.h.s. can be handled as we did for~\eqref{step 3}:
	\bml{ 
		\label{step 9}
		2 \eps \int_0^{\tj} \diff t \int_{s_n}^{\frac12 \lf( s_{n} + s_{n+1}\ri)} \diff s \: \chi_n(s) \ijj(t) \lf(i \partial_s u_n, \partial_t u_n \ri)	
		\geq - 2 \int_{\domj}\diff s \diff t \: \lf| \ijj(t) \ri| \lf| \eps \partial_s u_n\ri| \lf| \partial_t u_n \ri| \\
		\geq - C \eps\logi \int_{\domj}\diff s \diff t \: (1 - \eps \kj t) f_n^2 \lf[ \lf| \partial_t u_n \ri|^2 + \tx\frac{1}{(1 - 
			 \eps \kj t)^2} \lf| \eps \partial_s u_n \ri|^2 \ri],
	}
	where we have used~\eqref{eq:est log cost} with $k= k_n$, $k' = k_{n+1}$ and recalled that $|k_n-k_{n+1}| \leq C \eps$ to bound $\ijj$. 
	The last term in~\eqref{step 8} can be easily shown to provide a small correction: using~\eqref{eq:est log cost} again yields 
	$$ |\ijj(\tj)| \leq C \eps \logi f_n^2(\tj) ,$$
	so that by~\eqref{eq:Agmon decay} and \eqref{tj}
		\bml{
 			\label{eq:step 10}
 			\bigg| \int_{s_n}^{\frac12 \lf( s_{n} + s_{n+1}\ri)} \diff s \: \chi_n(s) \ijj(\tj) J_s[u_n](s,\tj) \bigg| \leq \int_{s_n}^{\frac12 \lf( s_{n} + s_{n+1}\ri)} \diff s \:  |\ijj(\tj)|  \lf| J_s[u_n] \ri| 	\\
			\leq C \eps \logi \int_{s_n}^{\frac12 \lf( s_{n} + s_{n+1}\ri)} \diff s \: f_n^2(\tj) \lf|u_n(s,\tj)\ri| \lf|\partial_s u_n(s,\tj) \ri|	\\
			\leq C \logi  \lf\| \eps \partial_s \psi \ri\|_{\infty} \int_{s_n}^{\frac12 \lf( s_{n} + s_{n+1}\ri)}\diff s \: \lf|\psi(s,\tj)\ri| = \OO(\eps^{\infty}),
		}
		where we have estimated the $s$-derivative of $ \psi $ by means of~\eqref{eq:sup est glm}.
		Hence, combining \eqref{step 8} with \eqref{step 9} and {\eqref{eq:step 10}}, we can bound from below~\eqref{step 7} as
		\bml{
 			\label{step 10}
 			\eps \int_0^{\tj} \diff t \: \lf[ F_n  - F_{n+1} \gjj^2 \ri] J_t[u_n](s_n,t) 	\\
			\geq \eps \int_0^{\tj} \diff t \int_{s_n}^{\frac12 \lf( s_{n} + s_{n+1}\ri)} \diff s \:\lf\{ - \partial_t \ijj J_s[u_n]	+ \partial_s \chi_n \ijj J_t[u_n] \ri\}	\\
			-  C \eps \logi \int_{\domj}\diff s \diff t \: (1 - \eps \kj t) f_n^2 \lf[ \lf| \partial_t u_n \ri|^2 + \tx\frac{1}{(1 - 
			 \eps \kj t)^2} \lf| \eps \partial_s u_n \ri|^2 \ri] + \OO(\eps^{\infty}).
		}
	To complete the proof it only remains to estimate the first two terms on the r.h.s. of the expression above, which again requires to borrow a bit of the kinetic energy. Using~\eqref{eq:est log der} we have 
	$$
		\sup_{t \in [0,\tj]} \bigg| \frac{\partial_t \ijj}{f_n^2} \bigg| \leq C \eps\logi,
	$$
	so that 
	\bml{
 		\label{step 11}
		\eps \bigg| \int_0^{\tj} \diff t \int_{s_n}^{\frac12 \lf( s_{n} + s_{n+1}\ri)} \diff s \: \partial_t \ijj J_s[u_n] \bigg| \leq C \eps \logi \int_{\domj} \diff s \diff t \: f_n^2 \lf| u_n \ri| \lf| \eps \partial_s u_n \ri| \\
		\leq C \eps \logi \int_{\domj}\diff s \diff t \: \lf[ \tx\frac{1}{\delta} \tx\frac{1}{1 - 
			 \eps \kj t} f_n^2 \lf| \eps \partial_s u_n \ri|^2 +  \delta |\psi|^2 \ri]	\\
		\leq C |\log\eps|^{-5} \int_{\domj}\diff s \diff t \: \tx\frac{1}{1 - 
			 \eps \kj t} f_n^2 \lf| \eps \partial_s u_n \ri|^2 + \OO(\eps^3\logi),
	}
	where we have chosen $ \delta = \eps |\log\eps|^a $, for some suitably large $ a > 0 $ to compensate the $ |\log\eps| $ prefactor (this generates the coefficient $ |\log\eps|^{-5} $), and used \eqref{eq:sup est glm} to estimate the remaining term.
	For the second term on the r.h.s. of~\eqref{step 10} we proceed in the same way, using first \eqref{eq:est log cost} and the assumption $ \lf| \partial_s \chi \ri| \leq C {\eps^{-1}} $, to get 
	\bml{
 		\label{step 12}
		\eps \bigg| \int_0^{\tj} \diff t \int_{s_n}^{\frac12 \lf( s_{n} + s_{n+1}\ri)} \diff s \: \partial_s \chi_n \ijj J_t[u_n] \bigg| \leq C {\eps} \logi \int_{\domj} \diff s \diff t \: f_n^2 \lf| u_n \ri| \lf| \partial_t u_n \ri| \\
		{\leq C \eps \logi \int_{\domj}\diff s \diff t \: \lf[ \tx\frac{1}{\delta}  f_n^2 \lf|\partial_t u_n \ri|^2 +  \delta |\psi|^2 \ri]}	\\
		{\leq C |\log\eps|^{-5} \int_{\domj}\diff s \diff t \:  f_n^2 \lf| \partial_t u_n \ri|^2 + \OO(\eps^3\logi),}	
	}
%
	where we have made the same choice of $ \delta $ as in \eqref{step 11}.
	
	Collecting all the previous estimates yields our claim~\eqref{eq:control boundary terms} (recall that there are $\neps \propto \eps ^{-1}$ terms to be summed, whence the final error of order $\eps ^2 \logi$).
	\end{proof}

\section{Density and Degree Estimates}\label{sec:density degree}

	In this section we prove the main results about the behavior of $ |\glm| $ close to the boundary of the sample $ \partial\Omega $ and an estimate of its degree at $ \partial \Omega $.
	
	We first notice that the $L^2$ estimate stated in \eqref{eq:main density} is in fact a trivial consequence of the energy asymptotics \eqref{eq:energy GL}: putting together the lower bounds~\eqref{eq:energy lb ann},~\eqref{eq:reduc func} and~\eqref{eq:bound reduc func} with the upper bound \eqref{eq:up bound GL}, we obtain
		\begin{equation}
			\label{eq:upper bound nonlinear}
			\frac{1}{2\eps b} \sum_{n=1}^{\neps} \int_{\cellj} \diff s \diff t \: \lf(1 - \eps k_n t\ri) f_n^4 \lf(1 - |u_n|^2 \ri)^2 \leq C \eps |\log\eps|^{\gamma}, 
		\end{equation}
	for some power $\gamma$ large enough (recall the meaning of the notation $\logi$). Now, using the fact that $ k_n = k(s) \lf(1 + \OO(\eps) \ri) $ inside $ \cellj $, we can easily reconstruct  \eqref{eq:main density}, once everything has been expressed in the original unscaled variables and the definitions \eqref{eq:ref profile} and \eqref{eq:splitting psi} has been exploited (recall also that $ \psi(s,t) = \glm(\rv(s,\eps t) $). See also \cite[Section 4.2]{CR2} for further details.

	\medskip
	
	We now focus on the refined density estimate discussed in Theorem \ref{theo:Pan} and the proof of Pan's conjecture. The result is obtained via an adaptation of the arguments used in \cite[Section 5.3]{CR2}, originating in~\cite{BBH1}. The general idea is now rather standard so we will mainly comment on the changes needed to make those argument work in the present setting. 
	
	\begin{proof}[Proof of Theorem \ref{theo:Pan}] The two main ingredients of the proof are the above estimate~\eqref{eq:upper bound nonlinear} and a pointwise bound on the gradient of $ u_n $. Once combined, the two estimates imply that the function $ |u_n| $ cannot be too far from $ 1 $ anywhere in the boundary layer $ \annd $ (see \eqref{eq:annd} for its precise definition).
	
	\medskip
	
%
		{\emph{Step 1, gradient estimate.}} A minor difference with the setting in \cite[Section 5.3]{CR2} is due to the convention we used to avoid a scaling of the tangential coordinate $ s $. This is just a matter of notation and by following \cite[Proof of Lemma 5.3]{CR2}, we can show that, for any $ n = 1, \ldots, \neps $,
		\beq
			\label{eq:point est grad u}
			\lf| \partial_t |u_n| \ri| \leq C f_n^{-1}(t) |\log\eps|^{3},	\qquad		\lf| \partial_s |u_n| \ri| \leq C  f_n^{-1}(t)  \eps^{-1}.
		\eeq
		Notice the second estimate above, which is a consequence of not scaling the coordinate $ s $.
		
		{We now prove~\eqref{eq:point est grad u}. From the definitions of $\psi$ and  $u_n$ we immediately have
		\bml{
			\lf| \dd_t |u_n|(s,t) \ri| \leq f_n^{-2}(t) \lf| f_n^{\prime}(t) \ri| \lf| \psi(s,t) \ri| + f_n^{-1}(t) \lf| \dd_t \lf| \psi(s,t) \ri| \ri|	\\
			\leq C f_n^{-1}(t) \lf[ |\log\eps|^3 + \lf| \dd_{t} \lf| \psi(s,t) \ri| \ri| \ri],
		}
		and 
		$$
			\lf| \dd_s |u_n|(s,t) \ri| \leq  f_n^{-1}(t) \lf| \dd_s \lf| \psi(s,t) \ri| \ri|
		$$
		where we have used~\cite[Equation (A.28)]{CR2}. The result is then a consequence of~\cite[Theorem~2.1]{Alm} or~\cite[Equation~(4.9)]{AH} in combination with the diamagnetic inequality (see \cite{LL}), which yield	
		\beq
			\lf| \nabla \lf| \glm \ri| \ri| \leq \lf| \lf( \nabla + i \tx\frac{\aavm}{\eps^2} \ri) \glm \ri|\leq C \eps ^{-1}
			\quad \Longrightarrow \quad \lf| \dd_{t} \lf| \psi(s,t) \ri| \ri| +  \eps \lf| \dd_{s} \lf| \psi(s,t) \ri| \ri|\leq  C.
		\eeq}
		
		\medskip
		
		{\emph{Step 2, uniform bound on $u_n$.}} {We first observe that the estimate $ \lf\| f_n - f_0 \ri\|_{\infty} = \OO(\eps)$ proven in~\eqref{eq:vari 1D opt density} guarantees that
		\beq
			f_n (t) \geq \game,	\qquad		\mbox{for any } (s,t) \in \cellj \cap \annd \mbox{ and } \forall n = 1, \ldots, \neps.
		\eeq}
		Now we can apply a standard argument to show that $ |u_n| $ can not differ too much from $ 1 $ {in $\annd$}. The proof is done by contradiction. {We choose some $ 0 < c < \frac{3}{2} a $ and define 
		\beq
			\sigme : = \eps^{1/4} \game^{-3/2} |\log\eps|^c \ll |\log\eps|^{c-3a/2} \ll 1.
		\eeq
		Suppose for contradiction that there exists a point $ (s_0,t_0) $ in $ \cellj \cap \annd $ such that
		\bdm
			\lf|1 - |u_n(s_0,t_0)| \ri| \geq \sigme.
		\edm}
		Then by \eqref{eq:point est grad u} we can construct a rectangle-like region $ R_{\eps} {\subset \cellj \cap \annd} $ of tangential length $ \frac{1}{2} \eps \game \sigme {\ll \eps} $ and normal length $ \varre $ with
		\beq
			\varre : = \game \sigme |\log\eps|^{-3} {\ll \eps^{1/6} |\log\eps|^{c-3-a/2} \ll |\log\eps|^{1/2}},
		\eeq
		{where}
		\bdm
			\lf|1 - |u_n(s,t)| \ri| \geq \tx\frac{1}{2} \sigme.
		\edm
		
		To complete the proof it suffices to estimate from below
		\bml{
		 	{\frac{1}{\eps} \sum_{n=1}^{\neps} \int_{\cellj} \diff s \diff t \: \lf(1 - \eps k_n t\ri) f_n^4 \lf(1 - |u_n|^2\ri)^2 \geq \frac{1}{\eps} \int_{R_{\eps}} \diff s \diff t \: \lf(1 - \eps k_n t\ri) f_n^4 \lf(1 - |u_n|^2\ri)^2}	\\
		 	 \geq \game^5 \sigme^3 \varre = \game^6 \sigme^4 |\log\eps|^{-3} = \eps |\log\eps|^{4c-3} \gg \eps |\log\eps|^{\gamma},
		}
		where $ \gamma $ is the power of $ |\log\eps |$ appearing in the r.h.s. of \eqref{eq:upper bound nonlinear} and we have chosen $ c$ so that $ c \geq \frac{1}{4}(\gamma+3) $. {Recalling the condition $ a > \frac{2}{3} c $ we also have $ a >  \frac{1}{6}(\gamma+3) ${, which coincides with the assumption on $ \game $ (see \eqref{eq:game})}. Under such conditions the estimate above contradicts the upper bound \eqref{eq:upper bound nonlinear} and the result is proven.}
		
		\medskip
		
		{\emph{Step 3, conclusion.} Now we know that in $\annd \cap \cellj$
		$$ \left||u_n| - 1\right| \leq \sigme, $$
		and it is easy to translate this estimate in an analogous one for $ |\psi(s,t)| $ and therefore $ |\glm| $. Indeed, in the cell $\cellj$ 
		$$ |\psi| = |\glm| = f_n |u_n|$$
		modulo a change of variables. The final estimate on $|\glm|$ then involves the reference profile $ \Gref $ but the bound $ \lf\| f_n - f_0 \ri\|_{\infty} = \OO(\eps)$ again allows the replacement of $ \Gref $ with $ f_0 $}. 
	\end{proof}	
	
	We can now turn to the proof of the estimate of the winding number of $ \glm $ along $ \partial \Omega $.
	
	 \begin{proof}[Proof of Theorem \ref{theo:circulation}] Thanks to the positivity of $ \Gref $ at $ t = 0 $ (see Lemma \ref{lem:point est fal}) and the result discussed above, $ \glm $ never vanishes on $ \partial \Omega $ and therefore its winding number is well defined. The rest of the proof follows the lines of \cite[Proof of Theorem 2.4]{CR2}. 
	 
	 {The first part is the estimate of the winding number contribution of the phase $ \phi_{\eps}$ involved in the change of gauge $ \psi(s,t) = \glm(\rv(s, \eps t)) e^{-i \phi_{\eps}(s,t)}$ but this can be done exaclty as in \cite[Proof of Lemma 5.4]{CR2}:
	 	\bml{
			\label{gauge phase}
			2 \pi \deg\lf(\glm, \partial \Omega\ri) - 2\pi \deg\lf(\psi, \partial \Omega \ri) = \int_{0}^{|\partial \Omega|} \diff s \: \gamma^{\prime}(s) \cdot \nabla \phi_{\eps}(s,t) =  \int_{0}^{|\partial \Omega|} \diff s \: \partial_{s} \phi_{\eps}(s,0)	\\
			= \phi_{\eps}(2\pi,0) - \phi_{\eps}(0,0) = \frac{1}{\eps^2} \int_{0}^{|\partial \Omega|} \diff s \: \gav^{\prime}(s) \cdot \aavm(\rv(s,0))	 - |\partial \Omega| \deps \\
			=  \frac{1}{\eps^2} \int_{\Omega} \diff \rv \: \curl \aavm - |\partial \Omega| \deps.
		}
		Now by the elliptic estimate \cite[Eq. (11.51)]{FH-book}
		\bdm
			 \lf\| \curl \aavm - 1 \ri\|_{C^{1}(\Omega)} = \OO(\eps),
		\edm
		and the Agmon estimate \cite[Eq. (12.10)]{FH1}
		\bdm
			\lf\| \nabla (\curl \aavm - 1) \ri\|_{L^1(\Omega\setminus \ann)} = \OO(\eps^{\infty}),
		\edm
		we get
		\bdn
			 \lf\| \curl \aavm - 1 \ri\|_{L^{1}(\ann)} & \leq & C \eps |\log\eps| \lf\| \nabla \lf( \curl \aavm -1\ri) \ri\|_{L^{\infty}(\Omega)} = \OO(\eps^2 |\log\eps|),	\nonumber	\\
			 \lf\| \curl \aavm - 1 \ri\|_{L^{1}(\Omega\setminus \ann)} & \leq & C  \lf\| \curl \aavm - 1 \ri\|_{L^{2}(\Omega\setminus \ann)} 	\nonumber \\
			 && \leq C \lf\| \nabla (\curl \aavm - 1) \ri\|_{L^1(\Omega\setminus \ann)} = \OO(\eps^{\infty}),
		\edn
		via Sobolev inequality. Altogether we can thus replace $ \curl \aavm $ with $ 1 $ in \eqref{gauge phase}, so obtaining
		\beq
			2 \pi \deg\lf(\glm, \partial \Omega\ri) - 2\pi \deg\lf(\psi, \partial \Omega \ri) =  \frac{|\Omega|}{\eps^2} + \OO(|\log\eps|).
		\eeq}
		
		A minor modification in the proof is then due to the cell decomposition and the use of a different decoupling in each cell: the analogue of \cite[Lemma~5.4]{CR2} is the following
	\beq
		\label{eq:circulation u}
	 	\sum_{n = 1}^{\neps}  \int_{s_n}^{s_{n+1}} \diff s \: J_s[u_n](s,0)  = \OO(|\log\eps|^{\infty}).
	\eeq
	To see that, we introduce a tangential cut-off function $ \chi(t) $ with support contained in $ [0,|\log\eps|^{-1}] $ and such that $ 0 \leq \chi \leq 1 $, $ \chi(0) = 1 $ and $ |\partial_t \chi|\ = \OO(|\log\eps|) $. Then we compute
	\bml{
	 	\int_{s_n}^{s_{n+1}} \diff s \: J_s[u_n](s,0)  =  \int_{s_n}^{s_{n+1}} \diff s \int_0^{\frac{1}{|\log\eps|}} \diff t \: \lf[ \partial_t \chi J_s[u_n](s,t) + \chi \partial_t J_s[u_n](s,t) \ri] = \\
	 	\int_0^{\frac{1}{|\log\eps|}} \diff t \bigg\{ \int_{s_n}^{s_{n+1}} \diff s \: \lf[ \partial_t \chi  J_s[u_n](s,t) + 2 \chi \lf(i \partial_t u_n, \partial_s u_n \ri) \ri] +  J_t[u_n](s_{n+1},t) -  J_t[u_n](s_{n},t) \bigg\}
	 	}
	 and after a rearrangement of the boundary terms
	 \bml{
	  	\label{circulation est 1}
	 	\sum_{n =1}^{\neps} \int_{s_n}^{s_{n+1}} \diff s \: J_s[u_n](s,0)  = \sum_{n =1}^{\neps} \int_0^{|\log\eps|^{-1}} \diff t \int_{s_n}^{s_{n+1}} \diff s \: \lf[ \lf( \partial_t \chi  \ri) J_s[u_n](s,t) + 2 \chi \lf(i \partial_t u_n, \partial_s u_n \ri) \ri]	\\
	 	 -  \sum_{n =1}^{\neps} \int_0^{|\log\eps|^{-1}} \diff t \: \lf[ J_t[u_{n+1}](s_{n+1},t) -  J_t[u_n](s_{n+1},t) \ri].
	 	 	}
	 The three terms on the r.h.s. of the above expression are going to be bounded independently. We first observe that, exactly like we derived \eqref{eq:upper bound nonlinear}, one can also extract from the comparison between the energy upper and lower bounds (see~\eqref{eq:bound reduc func}) the following estimate:
	 \beq
	 	\label{eq:upper bound kinetic}
	 	\sum_{n=1}^{\neps} \int_{\cellj} \diff s \diff t \: \lf(1 - \eps k_n t\ri) f_n^2 \lf\{ \lf| \partial_t u_n \ri|^2 + \tx\frac{1}{(1 - \eps k_n t)^2} \lf| \eps \partial_s u_n \ri|^2 \ri\} \leq C \eps^2 |\log\eps|^{\infty}.
	\eeq	
	Then we can estimate the absolute value of the first two terms on the r.h.s. of \eqref{circulation est 1} by using the Cauchy-Schwarz inequality 
	\bml{ 
		\sum_{n =1}^{\neps} \int_0^{|\log\eps|^{-1}} \diff t \int_{s_n}^{s_{n+1}} \diff s \: \lf[ C |\log\eps| |u_n| \lf| \partial_s u_n \ri| + 2 \lf| \partial_t u_n \ri| \lf| \partial_s u_n \ri| \ri]	\\
		\leq C \sum_{n =1}^{\neps} \int_{\cellj} \diff s \diff t \: (1 - \eps k_n t) f_n^2 \lf[|\log\eps| \lf| u_n \ri|^2 + \tx\frac{2}{(1 - \eps k_n t)^2} \lf| \partial_s u_n \ri|^2 + \lf| \partial_t u_n \ri|^2 \ri],
		}
		where we have exploited the pointwise lower bound \eqref{eq:fal point l u b}, which implies $ f_n(t) \geq C > 0 $ for any $ t \in [0,|\log\eps|^{-1}] $ and $  n = 1, \ldots, \neps $, to put back the density $ f_n^2 $ in the expression. Now the bound 
		$$ f_n |u_n| = |\psi| \leq 1 $$
		together with \eqref{eq:upper bound kinetic} yield
		\beq
			\sum_{n =1}^{\neps} \bigg| \int_0^{|\log\eps|^{-1}} \diff t \int_{s_n}^{s_{n+1}} \diff s \: \lf[ \lf( \partial_t \chi  \ri) J_s[u_n](s,t) + 2 \chi \lf(i \partial_t u_n, \partial_s u_n \ri) \ri] \bigg| \leq C |\log\eps|^{\infty}.
		\eeq
		On the other hand the definition \eqref{eq:splitting psi} of $ u_n $ implies that
		\bml{
		 	 \bigg| \int_0^{|\log\eps|^{-1}} \diff t \: J_t[u_{n+1}](s_{n+1},t) -  J_t[u_n](s_{n+1},t) \bigg| = \bigg| \int_0^{|\log\eps|^{-1}} \diff t \: \bigg( \frac{1}{f_{n+1}^2} - \frac{1}{f_n^2} \bigg) J_t[|\psi|](s_{n+1},t) \bigg|	\\
		 	 \leq C \eps |\log\eps|^{\infty} \int_0^{|\log\eps|^{-1}} \diff t \: |\psi| |\partial_s |\psi|| \leq C  |\log\eps|^{\infty},
		}
		thanks to \eqref{eq:vari 1D opt density}, the already mentioned lower bound on $ f_n $ in $ [0, |\log\eps|^{-1}] $ and the standard bound $ \lf\| \nabla \psi \ri\|_{\infty} \leq C \eps^{-1} $ (see, e.g., \cite[Eq. (11.50)]{FH-book}).
		
		Hence \eqref{eq:circulation u} is proven and the rest of the proof is just a repetition of the estimates in \cite[Proof of Theorem 2.4]{CR2}. Note that, as already anticipated in the comments after Theorem \ref{theo:circulation} $ \al_n = \al_0 (1 + \OO(\eps)) $, so that the optimal phases $ \al_n $ can all be replaced with $ \al_0 $.	
	\end{proof}

\appendix

\section{Useful Estimates on 1D Functionals}\label{sec:app}

Here we recall some preliminary results obtained in~\cite{CR2}. We start in Subsection~\ref{sec:app dens phase} with elementary properties of the minimizing 1D profiles and carry on in Subsection~\ref{sec:app cost} by recalling the crucial positivity property of the cost function we mentioned in Subsection~\ref{sec:sketch}.

\subsection{Properties of optimal phases and densities}\label{sec:app dens phase}

This subsection contains a summary of results on the 1D minimization problem that follow from relatively standard methods. We start with the well-posedness of the minimization problem at fixed $\alpha$. The following is~\cite[Proposition 3.1]{CR2}.

	\begin{pro}[\textbf{Optimal density $ \fkal $}]
		\label{pro:min fone}
		\mbox{}	\\
		For any given $ \alpha \in \R $, $ k \geq 0 $ and $ \eps$  small enough, there exists a minimizer $ \fkal $ to $ \fonekal$, unique up to sign, which we choose to be non-negative. It solves the variational equation
		\beq
			\label{eq:var eq fal}
			- \fkal^{\prime\prime} + \tx\frac{\eps k}{1 - \eps k t} \fkal^{\prime} + \potkal ( t) \fkal = \tx\frac{1}{\hex} \lf(1 - \fkal^2 \ri) \fkal
		\eeq
		with boundary conditions $ \fkal^{\prime}(0) =  \fkal^{\prime}(c_0|\log\eps|) = 0 $. Moreover $ \fkal $ satisfies the estimate
		\beq
			\label{eq:fal estimate}
			\lf\| \fkal \ri\|_{L^{\infty}(I_{\eps})} \leq 1
		\eeq
		and it is monotonically decreasing for $  t \geq \max \lf[ 0, -\alpha + \tx\frac{1}{\sqrt{\hex}} -C \eps \ri] $. 
		In addition $ \eonekal $ is a smooth function of $ \alpha \in \R $ and
		\beq
			\label{eq:eone explicit}
			\eonekal = - \frac{1}{2\hex} \int_{I_{\eps}} \diff  t \: (1 - \eps k  t) \fkal^4( t).
		\eeq
	\end{pro}

Next we consider the minimization problem as a function of the phase $\alpha$, dealt with in~\cite[Lemma~3.1]{CR2}. Here $\theo ^{-1}$ is defined as in~\eqref{eq:theo}. 
	
\begin{lem}[\textbf{Optimal phase $ \alk $}]
		\label{lem:opt phase}
		\mbox{}	\\
		For any $ 1 < \hex < \theo^{-1} $, $ k \geq 0 $ and $ \ep$ small enough, there exists at least one  $ \alk $ minimizing $ \eonekal $:
		\beq
			\label{eq:optimal energy}
			\inf_{\alpha \in \R} \eonekal = \eone_{k,\alk} =: \eonek.
		\eeq
		Setting $ \fk : = f_{k,\alk} $ we have that $ \fk > 0$ everywhere and 
		\beq
			\label{eq:FH nonlinear}
			\int_{I_{\eps}} \diff  t \: \frac{ t + \alk - \half \eps k t^2}{1 - \eps k t} \fk^2( t) = 0.
		\eeq
	\end{lem}

We also use some decay and gradient estimates for the minimizing density. The following is a combination of~\cite[Proposition 3.3 and Lemma A.1]{CR2}

	\begin{lem}[\textbf{Useful bounds on $ \fal $}]
		\label{lem:point est fal}
		\mbox{}	\\
		For any $ 1 < \hex < \theo^{-1} $, $ k \in \R $ and $ \eps $ sufficiently small, there exist two positive constants $ c, C > 0 $ independent of $ \eps $ such that
		\beq
			\label{eq:fal point l u b}
			c \: \exp\Big\{ - \tx\frac{1}{2}\big(  t + \sqrt{2} \big)^2 \Big\} \leq f_k( t) \leq C \: \exp\lf\{ - \tx\frac{1}{2} \lf(  t + \al \ri)^2 \ri\},
		\eeq
		for any $  t \in \ie $.
		
		Moreover there exists a finite constant $ C $ such that
		\beq
			\label{eq:fal derivative}
			\lf| f_k ^{\prime}( t) \ri| \leq C
			\begin{cases}
				1,	&	\mbox{for }  t \in \lf[ 0, |\alpha| + \tx\frac{2}{\sqrt{\hex}} \ri],	\\
				|\log\eps|^3 f_k ( t),		&	\mbox{for }  t \in \lf[ |\alpha| + \tx\frac{2}{\sqrt{\hex}}, c_0|\log\eps| \ri].
			\end{cases}
		\eeq
	\end{lem}

\subsection{Positivity of the cost function}\label{sec:app cost}

A less standard part of our analysis in~\cite{CR2} is the introduction of a cost function $\Kk$ whose positivity is one of the crucial ingredients of the energy lower bounds in the present paper.

Let us first recall the definition of the potential function associated with $\fk$:
\beq
\label{eq:Fk}
\Fk (t) : = 2 \int_0^t \diff \eta \: (1 - \eps k \eta) \fk^2(\eta) \frac{\eta + \alk - \half \eps k \eta^2}{(1 - \eps k \eta)^2},
\eeq
which has the following properties~\cite[Lemma 3.2]{CR2}:

\begin{lem}[\textbf{Properties of the potential function $\Fk$}]
		\label{lem:F prop}
		\mbox{}	\\
		For any $ 1 < \hex < \theo^{-1} $, $ k \in \R $ and $ \eps$ sufficiently small, we have
		\beq
			\label{F prop}
			\Fk(t) \leq 0,	\quad \mbox{in } \ieps,	\qquad \Fk(0) = \Fk(\teps) = 0.
		\eeq
	\end{lem}

The cost function that naturally enters the analysis is then 
\begin{equation}
		\label{eq:Kk}
		\Kk(t) = (1 -  \de) \fk^2(t) + \Fk(t)  
\end{equation}
	where $ \de $ is any parameter satisfying
	\beq
		\label{eq:de}
		0 < \de \leq C |\log\eps|^{-4},	\quad	\mbox{as } \eps \to 0.
	\eeq 

The positivity property we exploit is proved in~\cite[Proposition 3.5]{CR2}. Let 
	\beq
		\label{eq:annb}
		\annbk : = \lf\{ t \in \ieps \: : \fk\lf(t\ri) \geq |\log\eps|^3 \fk(\teps) \ri\},
	\eeq
	which is an interval in the $  t $ variable, i.e.,
	\beq
		\label{eq:annb bis}
		\annbk = [0, \btik],
	\eeq
	with
	\beq
		\label{eq:bxi}
		\btik \geq \teps - C \log|\log\eps|  = c_0 |\log\eps| \lf( 1 - \OO\lf( \tx\frac{\log|\log\eps|}{|\log\eps|} \ri) \ri).
	\eeq

	We then have

	\begin{lem}[\textbf{Positivity of the cost function}]
		\label{lem:K positive}
		\mbox{}	\\
		For any $ \de \in \R^+ $ satisfying~\eqref{eq:de}, $ 1 < \hex < \theo^{-1} $, $ k > 0 $ and $ \eps$ sufficiently small, we have
		\beq
			\label{eq:K positive}
			\Kk(t) \geq 0, \quad \mbox{ for any } t \in \annbk.
		\eeq
	\end{lem}

\end{document}